\documentclass{article}

\usepackage[english]{babel}

\usepackage[letterpaper,top=2cm,bottom=2cm,left=3cm,right=3cm,marginparwidth=1.75cm]{geometry}

\usepackage{amsmath}
\usepackage{graphicx}
\usepackage[colorlinks=true, allcolors=blue]{hyperref}
\usepackage{mathrsfs}
\usepackage{thm-restate}

\usepackage{algorithm}
\usepackage{algpseudocode}
\usepackage{calc}
\usepackage{amsthm}
\usepackage{amsfonts}
\usepackage{tikz}
\usetikzlibrary {arrows.meta, shapes}

\newcommand{\myindent}[1]{
\newline\makebox[#1cm]{}}
\newcommand*{\Let}[2]{\State #1 $\gets$
\parbox[t]{\linewidth-\algorithmicindent-\widthof{ #1 $\gets$}}{#2\strut}}
\newcommand*{\LongState}[1]{\State
\parbox[t]{\linewidth-\algorithmicindent}{#1\strut}}

\newtheorem{lemma}{Lemma}
\newtheorem{theorem}{Theorem}
\newtheorem{claim}{Claim}
\newtheorem{corollary}{Corollary}

\newtheorem{definition}{Definition}

\title{Deterministic Simple $(\Delta+\varepsilon\alpha)$-Edge-Coloring in\\Near-Linear Time}
\author{Michael Elkin\\
        \vspace{-0.15cm}\small Department of Computer Science,\\
        \vspace{-0.15cm}\small Ben-Gurion University of the Negev,\\
        \vspace{-0.1cm}\small Beer-Sheva, Israel.\\
        \href{elkinm@cs.bgu.ac.il}{elkinm@cs.bgu.ac.il}
        \and Ariel Khuzman\\
        \vspace{-0.15cm}\small Department of Computer Science,\\
        \vspace{-0.15cm}\small Ben-Gurion University of the Negev,\\
        \vspace{-0.1cm}\small Beer-Sheva, Israel.\\
        \href{huzmana@post.bgu.ac.il}{huzmana@post.bgu.ac.il}}
\date{}

\begin{document}
\maketitle

\begin{abstract}
We study the edge-coloring problem in simple $n$-vertex $m$-edge graphs with maximum degree $\Delta$. This is one of the most classical and fundamental graph-algorithmic problems. Vizing's celebrated theorem provides $(\Delta+1)$-edge-coloring in $O(m\cdot n)$ deterministic time. This running time was improved to $O\left(m\cdot\min\left\{\Delta\cdot\log n, \sqrt{n}\right\}\right)$~\cite{arjomandi1982efficient, gabow1985algorithms, sinnamon2019fast}, and very recently to randomized $\Tilde{O}\left(m\cdot n^{1/3}\right)~\cite{bhattacharya2024faster}$. It is also well-known that $3\left\lceil\frac{\Delta}{2}\right\rceil$-edge-coloring can be computed in $O(m\cdot\log\Delta)$ time deterministically~\cite{karloff1987efficient}. Duan et al.~\cite{duan2019dynamic} devised a randomized $(1+\varepsilon)\Delta$-edge-coloring algorithm with running time $O\left(m\cdot\frac{\log^6 n}{\varepsilon^2}\right)$. For the case of $\Delta=\left(\frac{\log n}{\varepsilon^4}\right)^{\Omega\left(\varepsilon^{-1}\cdot\log\varepsilon^{-1}\right)}$, Bhattacharya et al.~\cite{bhattacharya2024nibbling} devised a randomized $(1+\varepsilon)\cdot\Delta$-edge-coloring algorithm with running time $O\left(\frac{m\cdot\log\varepsilon^{-1}}{\varepsilon^2}\right)$. It was however open if there exists a deterministic near-linear time algorithm for this basic problem. We devise a simple deterministic $(1+\varepsilon)\Delta$-edge-coloring algorithm with running time $O\left(m\cdot\frac{\log n}{\varepsilon}\right)$. A randomized variant of our algorithm has running time $O(m\cdot(\varepsilon^{-18}+\log(\varepsilon\cdot\Delta)))$.

We also study edge-coloring of graphs with arboricity at most $\alpha$. Christiansen et al.~\cite{christiansen2023sparsity} recently devised a near-linear time algorithm for $(\Delta+2\alpha-2)$-edge-coloring. Bhattacharya et al.~\cite{bhattacharya2023density} devised a randomized $(\Delta+1)$-edge-coloring algorithm with running time  $\Tilde{O}(\min\{m\cdot\sqrt{n},m\cdot\Delta\}\cdot\frac{\alpha}{\Delta})$. Very recently Kowalik~\cite{kowalik2024edge} came up with a deterministic $(\Delta+1)$-edge-coloring algorithm with running time $O\left(m\cdot\alpha^7\cdot\log n\right)$. However, for large values of $\alpha$, the algorithms of~\cite{bhattacharya2023density, kowalik2024edge} require super-linear time. We devise a deterministic $(\Delta+\varepsilon\alpha)$-edge-coloring algorithm, for a parameter $\varepsilon>0$, with running time $O\left(\frac{m\cdot\log n}{\varepsilon^7}\right)$. A randomized version of our algorithm requires $O\left(\frac{m\cdot\log n}{\varepsilon}\right)$ expected time. Hence we show that the surplus in the number of colors can be reduced from $2\alpha$ of~\cite{christiansen2023sparsity} to $\varepsilon\alpha$ for an arbitrarily small $\varepsilon$, while still using near-linear time.

Our algorithm is based on a novel \emph{two-way degree-splitting}, which we devise in this paper. We believe that this technique is of independent interest.
\end{abstract}

\section{Introduction}
\subsection{General Graphs}
Given an undirected $n$-vertex $m$-edge graph, $G=(V,E)$, a function $\varphi:E\rightarrow [k]$, for a positive integer $k$, is called a \emph{proper $k$-edge-coloring} of $G$ if for any pair of incident distinct edges $e\neq e'$, we have $\varphi(e)\neq \varphi(e')$.
The smallest value of $k$ for which a proper $k$-edge-coloring of $G$ exists is called the \emph{chromatic index} of $G$, and is denoted by $\chi'(G)$.
Edge-coloring is a basic graph-theoretic primitive, and it is a staple in any course on Graph Theory. See, e.g., the monographs~\cite{bollobas1998modern},~\cite{jensen2011graph}, and the references therein. It was also extensively studied from algorithmic perspective~\cite{gabow1976using, gabow1982algorithms, arjomandi1982efficient, gabow1985algorithms, cole1982edge, karloff1987efficient, schrijver1998bipartite, cole2001edge, alon2003simple, sinnamon2019fast, duan2019dynamic, bernshteyn2023fast}.
There are also numerous applications of edge-coloring in networking and scheduling~\cite{williamson1997short, erlebach2001complexity, aggarwal2003switch, kodialam2003characterizing, cheng2007transmission, gandham2008link}.
Obviously, $\chi'(G)\geq\Delta$, where $\Delta=\Delta(G)$ is the maximum degree. A celebrated result by Vizing~\cite{vizing1964estimate} is that for any simple graph $G$, we have $\chi'(G)\leq\Delta+1$. The proof of Vizing is constructive, and gives rise to a deterministic $(\Delta+1)$-edge-coloring algorithm with running time $O(m\cdot n)$~\cite{misra1992constructive}. (See also~\cite{bollobas2012graph}, page 94.) This was improved to $O(m\cdot\min\{\Delta\cdot\log n, \sqrt{n\cdot\log n}\})$ by Arjomandi~\cite{arjomandi1982efficient} and Gabow et al.~\cite{gabow1985algorithms}, and further refined to $O(m\cdot\min\{\Delta\cdot\log n, \sqrt{n}\})$ by Sinnamon~\cite{sinnamon2019fast}. Bernshteyn and Dhawan~\cite{bernshteyn2023fast} recently devised a randomized algorithm for the problem with running time $O(n\cdot\Delta^{18})$, and Assadi~\cite{assadi2024faster} devised a randomized algorithm for this task that requires $O\left(n^2\cdot\log n\right)$ time. 
While these algorithms are efficient when $\Delta$ is reasonably small, this is no longer the case for larger values of $\Delta$. This led researchers to study algorithms that use more colors, but run in \emph{near-linear time}\footnote{By "near-linear time" we mean here $\Tilde{O}(m)$, where $\Tilde{O}(f(m))=O\left(f(m)\cdot\mathrm{poly}(\log(f(m)))\right)$.}. There are a number of simple deterministic $(2\Delta-1)$-edge-coloring algorithms known that require $O(m\cdot\log\Delta)$ time~\cite{bhattacharya2018dynamic, sinnamon2019fast}. It is also not hard to achieve a $3\left\lceil\frac{\Delta}{2}\right\rceil$-edge-coloring within the same time~\cite{karloff1987efficient, cole2001edge, duan2019dynamic}.
The bound $3\left\lceil\frac{\Delta}{2}\right\rceil$ is quite close to Shannon's bound of $\left\lfloor\frac{3\Delta}{2}\right\rfloor$~\cite{shannon1949theorem} for edge-coloring multigraphs. However, it is typically insufficient for \emph{simple} graphs, which are the focus of this paper.\\
Duan et al.~\cite{duan2019dynamic} devised the first near-linear time \emph{randomized} $(1+\varepsilon)\Delta$-edge-coloring algorithm, for any arbitrarily small constant $\varepsilon>0$. Specifically, assuming $\Delta=\Omega\left(\frac{\log n}{\varepsilon}\right)$, their algorithm requires $O\left(m\cdot\frac{\log^6 n}{\varepsilon^2}\right)$ time. For $\Delta=O\left(\frac{\log n}{\varepsilon}\right)$ one can achieve a deterministic near-linear time via the algorithm of Gabow et al.~\cite{gabow1985algorithms}, which requires $O(m\cdot\Delta\cdot\log n)=O\left(m\cdot\frac{\log^2 n}{\varepsilon}\right)$ time. Bhattacharya et al.~\cite{bhattacharya2024nibbling} further improved their bound for the case $\Delta\geq\Delta^*$, where $\Delta^*=\left(\frac{\log n}{\varepsilon^4}\right)^{\Theta\left(\varepsilon^{-1}\cdot\log\varepsilon^{-1}\right)}$. Specifically, their randomized algorithm computes $(1+\varepsilon)\Delta$-edge-coloring in this case in time $O\left(\frac{m\cdot\log\frac{1}{\varepsilon}}{\varepsilon^2}\right)$. This leaves open the question of whether there exists a \emph{deterministic} near-linear time algorithm for $(1+\varepsilon)\Delta$-edge-coloring problem, and also, whether the relatively high exponents of $\log n$ and of $\frac{1}{\varepsilon}$ in the running time of~\cite{duan2019dynamic} can be improved.\\
We answer these questions in the affirmative. First, we devise a simple deterministic $(1+\varepsilon)\Delta$-edge-coloring algorithm with running time $O\left(m\cdot\frac{\log n}{\varepsilon}\right)$. Second, a randomized variant of our algorithm employs the same number of colors, and requires $O\left(m\cdot\varepsilon^{-18}+m\cdot\log(\varepsilon\cdot\Delta)\right)$ time. See Table \ref{table: edge-coloring algorithms} for a concise summary of existing and new edge-coloring algorithms.

\begin{table}[h!]
    \begin{center}
    \addtolength{\leftskip} {-2cm}
    \addtolength{\rightskip}{-2cm}
    \begin{tabular}{|c | c | c | c|} 
     \hline
     Reference & $\#$ Colors & Running Time & Deterministic or Randomized \\ [0.5ex] 
     \hline
     \cite{bhattacharya2018dynamic},\cite{sinnamon2019fast} & $2\Delta-1$ & $O(m\cdot\log\Delta)$ & Deterministic \\ 
     \hline
     \cite{cole2001edge}+Upfal's reduction~\cite{karloff1987efficient} & $3\left\lceil\frac{\Delta}{2}\right\rceil$ & $O(m\cdot\log\Delta)$ & Deterministic \\
     \hline
     \cite{duan2019dynamic}+\cite{gabow1985algorithms} & $(1+\varepsilon)\Delta$ & $O\left(m\cdot\frac{\log^6 n}{\varepsilon^2}\right)$ & Randomized\\
     \hline
     \cite{bhattacharya2024nibbling} & $(1+\varepsilon)\Delta$ & $O\left(m\cdot\frac{\log \varepsilon^{-1}}{\varepsilon^2}\right)$ & Randomized$^{(*)}$ \\
     \hline
     \textbf{New} & $(1+\varepsilon)\Delta$ & $O\left(m\cdot\frac{\log n}{\varepsilon}\right)$ & Deterministic \\
     \hline
     \textbf{New} & $(1+\varepsilon)\Delta$ & $O\left(m\cdot\varepsilon^{-18}+m\cdot\log(\varepsilon\cdot\Delta)\right)$ & Randomized \\ [1ex] 
     \hline
    \end{tabular}
    \end{center}
    \caption{A concise summary of existing and new edge-coloring algorithms for general simple graphs that employs $(1+\varepsilon)\Delta$ colors or more.\\
    (*) The result of~\cite{bhattacharya2023density} is applicable only for $\Delta=\left(\frac{\log n}{\varepsilon^4}\right)^{\Omega\left(\varepsilon^{-1}\log\varepsilon^{-1}\right)}$.
    }
    \label{table: edge-coloring algorithms}
\end{table}

We note also that it is NP-hard to distinguish between graphs $G$ with $\chi'(G)=\Delta$ and $\chi'(G)=\Delta+1$~\cite{holyer1981np}. Our result can therefore be viewed as a deterministic near-linear time FPTAS for the edge-coloring problem.

\subsection{Graphs with Bounded Arboricity}
Edge-coloring graphs with bounded arboricity is being intensively studied in recent years~\cite{bhattacharya2023arboricity,bhattacharya2023density,christiansen2023sparsity,kowalik2024edge}. Given a graph $G=(V,E)$, its \emph{arboricity} $\alpha=\alpha(G)=\max_{{{S\subseteq V,}\atop |S|\geq 2}}\left\lceil\frac{|E(S)|}{|V(S)|-1}\right\rceil$. The family of graphs with bounded arboricity is quite a rich graph-family: in particular, it contains graphs with bounded genus, graphs that exclude any fixed minor, and graphs with bounded tree-width. Bhattacharya et al.~\cite{bhattacharya2023density} devised a randomized algorithm that computes a $(\Delta+1)$-edge-coloring in $\Tilde{O}\left(\min\{m\sqrt{n},m\cdot\Delta\}\cdot\frac{\alpha}{\Delta}\right)$ time, improving the bounds of~\cite{gabow1985algorithms, sinnamon2019fast} wherever $\alpha=o(\Delta)$. Christiansen et al.~\cite{christiansen2023sparsity} devised a deterministic $(\Delta+2\alpha-2)$-edge-coloring algorithm with running time $O(m\cdot\log\Delta)$. (Their algorithm can, in fact, be viewed as an adaptation of distributed $(\Delta+O(\alpha))$-edge-coloring algorithm due to Barenboim et al.~\cite{barenboim2017deterministic} to sequential setting, in which each $H$-set of the $H$-decomposition is just a single vertex.) Recently, Kowalik~\cite{kowalik2024edge} devised a deterministic $(\Delta+1)$-edge-coloring algorithm with running time $O(m\cdot\alpha^7\cdot\log n)$. Also, under a very mild condition on $\Delta$ and $\alpha$ ($\Delta\geq 4\alpha$ is stronger than the condition in~\cite{kowalik2024edge}), Kowalik~\cite{kowalik2024edge} devised a deterministic $\Delta$-edge-coloring algorithm with the same running time, and a randomized $\Delta$-edge-coloring algorithm with running time $O(m\cdot\alpha^3\cdot\log n)$. (Some condition on $\Delta$ and $\alpha$ is necessary for $\chi'(G)=\Delta$; see~\cite{kowalik2024edge} for more details, and for an excellent survey of further relevant literature.)

We note, however, that for large values of $\alpha$, the algorithms of~\cite{bhattacharya2023density} and~\cite{kowalik2024edge} require super-linear time. On the other hand, the algorithm of~\cite{christiansen2023sparsity} uses near-linear time, but employs $\Delta+2\alpha-2$ colors, i.e., the surplus in the number of colors (with respect to the desired bound of $\Delta$ or $\Delta+1$) is $\approx 2\alpha$. We show that this surplus can be significantly decreased, while still using deterministic near-linear time. Specifically, for any parameter $\varepsilon>0$, our deterministic $(\Delta+\varepsilon\alpha)$-edge-coloring algorithm requires $O\left(\frac{m\cdot\log n}{\varepsilon^7}\right)$ time. Moreover, a randomized variant of this algorithm requires expected $O\left(\frac{m\cdot\log n}{\varepsilon}\right)$ time. See Table \ref{table: arboricity edge-coloring algorithms} for a concise summary of the most relevant existing and new edge-coloring algorithms for graphs with bounded arboricity.

\begin{table}[h!]
    \begin{center}
    \addtolength{\leftskip} {-2cm}
    \addtolength{\rightskip}{-2cm}
    \begin{tabular}{|c | c | c | c|} 
     \hline
     Reference & $\#$ Colors & Running Time & Deterministic or Randomized \\ [0.5ex] 
     \hline
     \cite{kowalik2024edge} & $\Delta$ & $O(m\cdot\alpha^7\log n)$ & Deterministic$^{(**)}$ \\ 
     \hline
     \cite{kowalik2024edge} & $\Delta$ & $O(m\cdot\alpha^3\log n)$ & Randomized$^{(**)}$ \\ 
     \hline
     \cite{bhattacharya2023density} & $\Delta+1$ & $\Tilde{O}(\min\{m\sqrt{n},m\cdot\Delta\}\cdot\frac{\alpha}{\Delta})$ & Randomized \\
     \hline
     \cite{kowalik2024edge} & $\Delta+1$ & $O\left(m\cdot\alpha^7\cdot\log n\right)$ & Deterministic\\
     \hline
     \cite{christiansen2023sparsity},\cite{barenboim2017deterministic} & $\Delta+2\alpha-2$ & $O\left(m\cdot\log\Delta\right)$ & Deterministic \\
     \hline
     \textbf{New} & $\Delta+\varepsilon\alpha$ & $O\left(m\cdot\frac{\log n}{\varepsilon^7}\right)$ & Deterministic \\
     \hline
     \textbf{New} & $\Delta+\varepsilon\alpha$ & $O\left(m\cdot\frac{\log n}{\varepsilon}\right)$ & Randomized \\ [1ex] 
     \hline
    \end{tabular}
    \end{center}
    \caption{A concise summary of most relevant existing and new results on edge-coloring graphs with arboricity at most $\alpha$.\\
    (**) The result holds if $\Delta\geq 4\alpha$.
    See~\cite{kowalik2024edge} for more details, and for additional results about this problem.
    }
    \label{table: arboricity edge-coloring algorithms}
\end{table}

\subsection{Technical Overview}
Many known $(\Delta+1)$-edge-coloring algorithms~\cite{gabow1976using, arjomandi1982efficient, gabow1985algorithms, sinnamon2019fast} are based on \emph{degree-splitting}. A degree-splitting $(E_1,E_2)$ of a graph $G=(V,E)$ is a partition of the edge set $E$ of $G$ into blue ($E_1$) and red ($E_2$) edges, in such a way that for every vertex $v\in V$, the number of blue edges incident to it is almost equal to the number of red edges that are incident to it. Such degree-splittings can be efficiently computed in near-linear time. The algorithms of~\cite{arjomandi1982efficient, gabow1985algorithms, sinnamon2019fast} compute such a degree-splitting, recurse on $E_1$ and $E_2$, and then fix the resulting coloring, so that it will use just $\Delta+1$ colors.

Our first observation is that if one just omits the fixing step, the algorithm still produces $(1+\varepsilon)\Delta$-edge-coloring, and does so in near-linear time.

For graphs with bounded arboricity, the algorithm of ~\cite{christiansen2023sparsity} (based on~\cite{barenboim2017deterministic}) orders the vertices $v_1,v_2,...,v_n$ in the degeneracy-order, i.e., so that for every vertex $v_i$, $i\in\{1,2,...,n\}$, the number of neighbors of $v_j$ with $j>i$ that it has is at most $d$, where $d$ is the degeneracy of $G$. (The \emph{degeneracy} $d$ of $G$ is the minimum number such that one can order vertices as above, where the number of higher-index neighbors is at most $d$. It is well-known that $d\leq2\alpha-1$.) The algorithm then processes the vertices one after another, starting with $v_n$ and ending with $v_1$. For every vertex $v_i$ it colors its incident edges $(v_i,v_j)$, $j>i$, using the palette $1,2,...,\Delta+d-1\leq\Delta+2\alpha-2$.

To devise $(\Delta+\varepsilon\alpha)$-edge-coloring in near-linear time, we develop a \emph{two-way degree-splitting}. Specifically, consider a graph $G=(V,E)$ with an orientation $\mu$ on all edges. We show that one can efficiently partition its edge set into blue and red edges $(E_1,E_2)$, so that for every vertex $v$, not only the \emph{total} numbers of blue and red edges incident on $v$ are almost equal, but also the numbers of outgoing (under $\mu$) blue and red edges incident on $v$ are almost equal as well. (This is, therefore, obviously the case for incoming edges too.)

Given an efficient procedure for computing a two-way degree-splitting, one can compute a $(\Delta+\varepsilon\alpha)$-edge-coloring by a similar recursion to the one that we use for $(1+\varepsilon)\Delta$-edge-coloring. We start with computing a forests-decomposition of the input graph into $2\alpha$ forests, and it induces an orientation $\mu$ with maximum outdegree at most $2\alpha$~\cite{barenboim2008sublogarithmic}. (In fact, one can use here a more recent forests-decomposition of~\cite{blumenstock2020constructive} into $(1+\varepsilon)\alpha$ forests, for an arbitrarily small $\varepsilon>0$.) We then compute the two-way degree-splitting $(E_1,E_2)$, with respect to the orientation $\mu$, recurse on both parts, and continue the recursion up to the stage where the input subgraphs have small arboricities. On these subgraphs we invoke a $(\Delta+1)$-edge-coloring algorithm for graphs of small arboricity~\cite{bhattacharya2023density, kowalik2024edge}. The running time of these algorithms grows with the arboricity of their input subgraphs, but the latter is guaranteed to be small. As a result, the total running time of invoking these algorithms on each of the (edge-disjoint) small-arboricity subgraphs produced by our algorithm is near-linear in the number of edges $|E|$ of the original input graph $G=(V,E)$. The depth of the recursion is roughly $\log\alpha$, and thus the overall time required to compute all the two-way degree-splittings throughout the recursion is also near-linear in $|E|$. Moreover, note that in the algorithms of ~\cite{arjomandi1982efficient, gabow1985algorithms,sinnamon2019fast} and in our $(1+\varepsilon)\Delta$-edge-coloring algorithm, the recursion depth is roughly $\log\Delta$, and as a result, the surplus in the number of employed colors ($\varepsilon\Delta$) is linear in $\Delta$. On the other hand, when we use two-way degree-splitting, the recursion depth is just $\approx\log\alpha$, and as a result the number of extra colors employed by our algorithm is just $\varepsilon\alpha$. We remark that unlike the algorithms of~\cite{gabow1985algorithms,sinnamon2019fast}, our algorithms (both the one that employs an ordinary degree-splitting and the one that employs the novel two-way degree-splitting) do the computationally-expensive edge-recoloring step (sometimes it is also called "fixing step") only on the bottom level level of the recursion.

The heart of our algorithm is, therefore, the computation of the two-way degree-splitting. Here we build upon classical works of Israeli and Shiloach~\cite{israeli1986improved} and Liang et al.~\cite{liang1996parallel} on ordinary degree-splitting. (We remark that Gabow et al.~\cite{gabow1976using,gabow1985algorithms}, Arjomandi~\cite{arjomandi1982efficient} and Karloff and Shmoys~\cite{karloff1987efficient} also studied ordinary (one-way) degree-splittings. Those, however, are based on Eulerian cycles, while our method is closer to that of~\cite{israeli1986improved,liang1996parallel}.)
The algorithm of~\cite{israeli1986improved} (for parallel maximal matching) repeatedly extracts maximal paths from the input graph, and colors them alternately by red and blue. (Intuitively, a path is \emph{maximal} if it cannot be extended.)
We introduce the notion of \emph{alternating-directions} (or AD) paths, as paths $(v_0,e_1,v_1,e_2,v_2,...,e_k,v_k)$, in which (essentially) all pairs of consecutive edges $e_i,e_{i+1}$ are oriented (under a given edge-orientation $\mu$) in opposite directions. See Figure \ref{fig:alternating-directions path} for an illustration.

Our algorithm has a similar structure to that of~\cite{israeli1986improved}, as it repeatedly extracts maximal AD paths and colors them alternately. Some discrepancy (between the number of red and blue edges incident on a given vertex $v$) is induced only by endpoints of these paths. However, their maximality implies that no vertex can serve as an endpoint of too many such paths.

We believe that the two-way degree-splittings and AD paths that we introduce in this paper are of independent interest, and will be found useful for efficient computation of edge-coloring and maximal matchings in graphs of bounded arboricity in various computational settings. In particular, we consider the dynamic edge-coloring problem to be a likely candidate to benefit from this new technique.

\subsection{Related Work}

Algorithms for $\Delta$-edge-coloring bipartite graphs were provided in~\cite{gabow1976using, cole1982edge, gabow1982algorithms, schrijver1998bipartite, cole2001edge, alon2003simple}. The problem of efficient edge-coloring of multigraphs was studied in~\cite{hochbaum1986better, nishizeki19901, sanders2005asymptotic}. Efficient algorithms for edge-coloring planar graphs were devised in~\cite{chrobak1990improved, cole2008new}. Edge-coloring of pseudo-random graphs (i.e., graphs with a bounded second-largest eigenvalue  of their adjacency matrix) was studied in~\cite{ferber2020towards}. The edge-coloring problem was also extensively studied in the context of distributed algorithms~\cite{barenboim2011distributed, barenboim2013basic, elkin20142delta, barenboim2017deterministic, balliu2022distributed, ghaffari2020improved, balliu2020distributed, bernshteyn2022fast}, parallel algorithms~\cite{karloff1987efficient, chrobak1989fast, liang1995fast, liang1996parallel, grable1997nearly} and online algorithms~\cite{bahmani2012online, bhattacharya2021online, kulkarni2022online, bar1992greedy}.

Ordinary (one-way) degree-splittings in distributed model were studied in~\cite{hanckowiak2001distributed,barenboim2017deterministic,ghaffari2018deterministic}. The work of Barenboim et al.~\cite{barenboim2017deterministic} is probably the closest in spirit to our work, as it considers distributed edge-coloring graphs with arboricity $\alpha$ using $\Delta+O(\alpha)$ and more colors. While their basic algorithm (that builds ($\Delta+O(\alpha)$)-edge-coloring) does not employ degree-splitting, they also devised a sequence of (distributed) algorithms that use $\Delta+O\left(\sqrt{\Delta\alpha}\right)$, $\Delta+O\left(\Delta^{\frac{2}{3}}\alpha^{\frac{1}{2}}\right)$,... colors, and these algorithms implicitly employ two-way degree-splittings. These splittings are based on defective edge-colorings (see~\cite{barenboim2011distributed, barenboim2013basic}), and result in large discrepancies. As a result, to the best of our understanding, they cannot be used for edge-coloring with surplus smaller than $2\alpha-O(1)$.

\subsection{Structure of the Paper}
Section \ref{sec: Efficient (1+varepsilon)Delta-Edge-Coloring} contains our degree-dependent edge-coloring algorithm. It can be viewed as a warm-up to our main arboricity-dependent algorithm, provided in Section \ref{sec: Arboricity-Dependent Edge-Coloring}. As a first part, in Section \ref{sec: The Oriented Degree-Splitting Problem}, we present the oriented degree-splitting problem, and provide an algorithm for computing a two-way (or oriented) degree-splitting. We use this algorithm in Section \ref{sec: An Arboricity-Dependent Edge-Coloring Algorithm}, where we describe the main arboricity-dependent edge-coloring. In Appendix \ref{app: Forests-Decomposition Orientation} we provide a well-known algorithm for computing a forests decomposition orientation, that is used in the arboricity-dependent edge-coloring algorithm. Some proofs from Section \ref{sec: Efficient (1+varepsilon)Delta-Edge-Coloring} can be found in Appendix \ref{App: Some Proofs from Section 3}.

\section{Preliminaries}\label{sec preliminaries}

Let $G=(V,E)$ be a graph. For a vertex $v\in V$, denote the degree of $v$ in $G$ by $\deg_G(v)$.\\
We denote $|V|=n$, $|E|=m$ and the maximum degree of $G$ by $\Delta(G)=\max_{v\in V}\deg(v)$ (or $\Delta$, for short).\\
For a directed graph, we denote by $\text{indeg}(v)$ the incoming degree of a vertex $v\in V$, that is, $\text{indeg}(v)=\left|\left\{\langle u,v\rangle\in E\right\}\right|$. We also denote by $\text{outdeg}(v)$ the outgoing degree of a vertex $v\in V$, that is, $\text{outdeg}(v)=\left|\left\{\langle v,u\rangle\in E\right\}\right|$.
\begin{definition}[Eulerian graph]
    A graph $G=(V,E)$ is called Eulerian graph if and only if all its vertices have even degrees.
\end{definition}

\begin{definition}[Adjacent edges]
    Given a graph $G=(V,E)$, we say that two edges $e,e'\in E$ are adjacent, or neighbors of each other, if $e\neq f$ and they share an endpoint.
\end{definition}

\begin{definition}[Proper edge-coloring]
    A proper $k$-edge-coloring of a graph $G=(V,E)$ is a map $\varphi:E\rightarrow\{1,2,...,k\}$, such that $\varphi(e)\neq\varphi(e')$ for every pair of adjacent edges $e,e'$.
\end{definition}

We will assume without loss of generality that there are no isolated vertices in the graph, i.e, $m\geq \frac{n}{2}$. Otherwise, we can omit the isolated vertices in the graph in each level of recursion.

\section{Efficient $(1+\varepsilon)\Delta$-Edge-Coloring}\label{sec: Efficient (1+varepsilon)Delta-Edge-Coloring}

In this section we present our edge-coloring algorithm, and analyse it.\\
Before describing the algorithm, we first state two results for $(\Delta+1)$-edge-coloring in deterministic and randomized settings, which we use in blackbox manner.

\begin{theorem}[$(\Delta+1)$-edge-coloring]\label{Delta+1 col}
\begin{enumerate}
    \item[(1)] ~\cite{gabow1985algorithms} There is a deterministic algorithm that given a graph $G$ with $n$ vertices, $m$ edges, and maximum degree $\Delta$, outputs a proper $(\Delta+1)$-edge-coloring of $G$ in $O(m\cdot\Delta\cdot\log n)$ time.
    \item[(2)] ~\cite{bernshteyn2023fast} There is a randomized algorithm that given a graph $G$ with $n$ vertices and maximum degree $\Delta$, outputs a proper $(\Delta+1)$-edge-coloring of $G$ in $O\left(\Delta^{18}\cdot n\right)$ time with probability at least $1-\frac{1}{\Delta^n}$.
\end{enumerate}
\end{theorem}

\subsection{The Algorithm}
In order to compute the edge-coloring in our algorithms we use a degree-splitting method. We start by presenting the Degree-Splitting problem.\\

\textbf{The Degree-Splitting Problem.}
The undirected degree-splitting problem seeks a partitioning of the graph edges $E$ into two parts so that the partition looks almost balanced around each vertex. Concretely, we need to color each edge red or blue such that for each node, the difference between the number of red and blue edges incident on it is at most some small discrepancy threshold value $\kappa$. In other words, we want an assignment $q:E \rightarrow \{1, -1\}$ such that for each node $v\in V$, we have $$\left|\sum_{e\in E_v} q(e)\right|\leq \kappa,$$
where $E_v$ denotes the set of edges incident on $v$.\\

Before we describe our edge-coloring algorithm, we present an algorithm that computes a degree-splitting with discrepancy at most 2 that we will use in our main edge-coloring algorithm. The algorithm is a special case of a result due to Beck and Fiala~\cite{beck1981integer}, who showed that any hypergraph of rank $t$ (each
hyperedge has at most $t$ vertices) admits an edge-coloring in two colors with discrepancy at most $2t-2$. The degree-splitting algorithm is also used in previous works~\cite{israeli1986improved, ghaffari2020improved} in $\mathrm{PRAM}$ and distributed settings.
The idea of the algorithm for the case of simple graphs ($t=2$) is as follows:\\
For now we will assume that the input graph has an even number of edges. For a connected input graph $G=(V,E)$ with an even number of edges, we add a dummy node to $G$, and connect it to all the odd degree vertices of the graph to obtain a graph $G'$. In this way all the degrees of $G'$ are even, and it is possible to compute an Eulerian cycle in this graph efficiently. Observe that this cycle is of even length. So we can color the edges of this cycle in two alternating colors and then by deleting the dummy node and its incident edges we will receive the desired coloring of the edges of $G$ with discrepancy at most 1. 
In the case that the input graph has an odd number of edges, we will analyze this algorithm more carefully and achieve a coloring of the edges of $G$ with discrepancy at most 2.

The description of this algorithm is given in Theorem \ref{Procedure Degree-Splitting}. See also Figure \ref{fig:deg-splitting} for an illustration of the degree-splitting procedure. Its proof is quite standard. For completeness we provide it in Appendix \ref{App: Some Proofs from Section 3}.

\begin{theorem}[Procedure \textsc{Degree-Splitting}]\label{Procedure Degree-Splitting}
    Let $G=(V,E)$ be a graph with $m$ edges. A degree-splitting of $G$ with discrepancy $\kappa=2$ can be computed in $O(m)$ time. Moreover, there is almost the same amount of blue and red edges. Namely, there are $\left\lfloor\frac{m}{2}\right\rfloor$ edges that are colored using one of the colors, and $\left\lceil\frac{m}{2}\right\rceil$ edges that are colored using the second color.
\end{theorem}

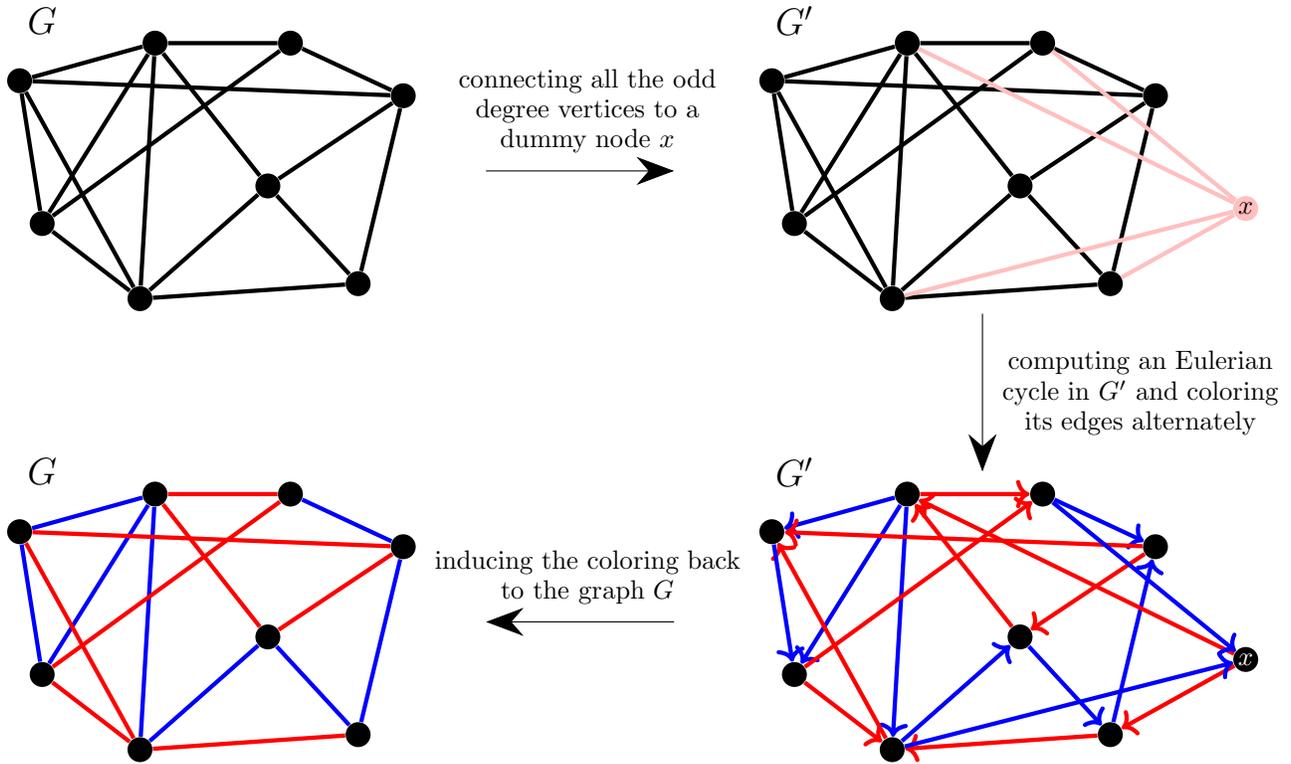
\begin{figure}
    \begin{center}
    \addtolength{\leftskip} {-2cm} 
    \addtolength{\rightskip}{-2cm}
    \begin{tikzpicture}
    \begin{scope}[xshift=-5cm]
    
        \node[circle, fill=black] at (0.3,1.3) (n1) {};
        \node[circle, fill=black] at (-1,2.3) (n2) {};
        \node[circle, fill=black] at (2,2.8) (n3) {};
        \node[circle, fill=black] at (0.5,4.7) (n4) {};
        \node[circle, fill=black] at (3.8,4) (n5) {};
        \node[circle, fill=black] at (-1.3,4.2) (n6) {};
        \node[circle, fill=black] at (2.3,4.7) (n7) {};
        \node[circle, fill=black] at (3.2,1.5) (n8) {};

        \path [-](n1) edge[ultra thick] (n2);
        \path [-](n1) edge[ultra thick] (n4);
        \path [-](n1) edge[ultra thick] (n3);
        \path [-](n2) edge[ultra thick] (n4);
        \path [-](n3) edge[ultra thick] (n5);
        \path [-](n4) edge[ultra thick] (n7);
        \path [-](n4) edge[ultra thick] (n6);
        \path [-](n5) edge[ultra thick] (n7);
        \path [-](n3) edge[ultra thick] (n4);
        \path [-](n1) edge[ultra thick] (n6);
        \path [-](n3) edge[ultra thick] (n8);
        \path [-](n5) edge[ultra thick] (n8);
        \path [-](n1) edge[ultra thick] (n8);
        \path [-](n2) edge[ultra thick] (n7);
        \path [-](n2) edge[ultra thick] (n6);
        \path [-](n5) edge[ultra thick] (n6);
    
    \end{scope}

    \begin{scope}[xshift=5cm]
    
        \node[circle, fill=black] at (0.3,1.3) (n1) {};
        \node[circle, fill=black] at (-1,2.3) (n2) {};
        \node[circle, fill=black] at (2,2.8) (n3) {};
        \node[circle, fill=black] at (0.5,4.7) (n4) {};
        \node[circle, fill=black] at (3.8,4) (n5) {};
        \node[circle, fill=black] at (-1.3,4.2) (n6) {};
        \node[circle, fill=black] at (2.3,4.7) (n7) {};
        \node[circle, fill=black] at (3.2,1.5) (n8) {};
        \node[circle, fill=pink] at (5,2.5) (n9) {};

        \path [-](n1) edge[ultra thick] (n2);
        \path [-](n1) edge[ultra thick] (n4);
        \path [-](n1) edge[ultra thick] (n3);
        \path [-](n2) edge[ultra thick] (n4);
        \path [-](n3) edge[ultra thick] (n5);
        \path [-](n4) edge[ultra thick] (n7);
        \path [-](n4) edge[ultra thick] (n6);
        \path [-](n5) edge[ultra thick] (n7);
        \path [-](n3) edge[ultra thick] (n4);
        \path [-](n1) edge[ultra thick] (n6);
        \path [-](n3) edge[ultra thick] (n8);
        \path [-](n5) edge[ultra thick] (n8);
        \path [-](n1) edge[ultra thick] (n8);
        \path [-](n2) edge[ultra thick] (n7);
        \path [-](n2) edge[ultra thick] (n6);
        \path [-](n5) edge[ultra thick] (n6);
        \path [-](n1) edge[ultra thick, pink] (n9);
        \path [-](n4) edge[ultra thick, pink] (n9);
        \path [-](n7) edge[ultra thick, pink] (n9);
        \path [-](n8) edge[ultra thick, pink] (n9);
    
    \end{scope}
    \begin{scope}[xshift=5cm, yshift=-6cm]
    
        \node[circle, fill=black] at (0.3,1.3) (n1) {};
        \node[circle, fill=black] at (-1,2.3) (n2) {};
        \node[circle, fill=black] at (2,2.8) (n3) {};
        \node[circle, fill=black] at (0.5,4.7) (n4) {};
        \node[circle, fill=black] at (3.8,4) (n5) {};
        \node[circle, fill=black] at (-1.3,4.2) (n6) {};
        \node[circle, fill=black] at (2.3,4.7) (n7) {};
        \node[circle, fill=black] at (3.2,1.5) (n8) {};
        \node[circle, fill=black] at (5,2.5) (n9) {};

        \path [->](n2) edge[ultra thick, red] (n1);
        \path [->](n4) edge[ultra thick, blue] (n1);
        \path [->](n1) edge[ultra thick, blue] (n3);
        \path [->](n4) edge[ultra thick, blue] (n2);
        \path [->](n5) edge[ultra thick, red] (n3);
        \path [->](n4) edge[ultra thick, red] (n7);
        \path [->](n4) edge[ultra thick, blue] (n6);
        \path [->](n7) edge[ultra thick, blue] (n5);
        \path [->](n3) edge[ultra thick, red] (n4);
        \path [->](n1) edge[ultra thick, red] (n6);
        \path [->](n3) edge[ultra thick, blue] (n8);
        \path [->](n8) edge[ultra thick, blue] (n5);
        \path [->](n8) edge[ultra thick, red] (n1);
        \path [->](n2) edge[ultra thick, red] (n7);
        \path [->](n6) edge[ultra thick, blue] (n2);
        \path [->](n5) edge[ultra thick, red] (n6);
        \path [->](n9) edge[ultra thick, red] (n4);
        \path [->](n9) edge[ultra thick, red] (n8);
        \path [->](n1) edge[ultra thick, blue] (n9);
        \path [->](n7) edge[ultra thick, blue] (n9);
    
    \end{scope}

    \begin{scope}[xshift=-5cm, yshift=-6cm]
    
        \node[circle, fill=black] at (0.3,1.3) (n1) {};
        \node[circle, fill=black] at (-1,2.3) (n2) {};
        \node[circle, fill=black] at (2,2.8) (n3) {};
        \node[circle, fill=black] at (0.5,4.7) (n4) {};
        \node[circle, fill=black] at (3.8,4) (n5) {};
        \node[circle, fill=black] at (-1.3,4.2) (n6) {};
        \node[circle, fill=black] at (2.3,4.7) (n7) {};
        \node[circle, fill=black] at (3.2,1.5) (n8) {};

        \path [-](n2) edge[ultra thick, red] (n1);
        \path [-](n4) edge[ultra thick, blue] (n1);
        \path [-](n1) edge[ultra thick, blue] (n3);
        \path [-](n4) edge[ultra thick, blue] (n2);
        \path [-](n5) edge[ultra thick, red] (n3);
        \path [-](n4) edge[ultra thick, red] (n7);
        \path [-](n4) edge[ultra thick, blue] (n6);
        \path [-](n7) edge[ultra thick, blue] (n5);
        \path [-](n3) edge[ultra thick, red] (n4);
        \path [-](n1) edge[ultra thick, red] (n6);
        \path [-](n3) edge[ultra thick, blue] (n8);
        \path [-](n8) edge[ultra thick, blue] (n5);
        \path [-](n8) edge[ultra thick, red] (n1);
        \path [-](n2) edge[ultra thick, red] (n7);
        \path [-](n6) edge[ultra thick, blue] (n2);
        \path [-](n5) edge[ultra thick, red] (n6);
    
    \end{scope}

    \draw [-{Stealth[length=5mm]}] (-0.1,3) -- (2.4,3);
    \node at (1.25,4.2){connecting all the odd};
    \node at (1.25,3.8){degree vertices to a};
    \node at (1.25,3.4){dummy node $x$};
    \draw [-{Stealth[length=5mm]}] (2.4,-3) -- (-0.1,-3);
    \node at (1.25,-2.2){inducing the coloring back};
    \node at (1.25,-2.6){to the graph $G$};
    \draw [-{Stealth[length=5mm]}] (6.5,1.1) -- (6.5,-1);
    \node at (8.6,0.45){computing an Eulerian};
    \node at (8.6,0.05){cycle in $G'$ and coloring};
    \node at (8.6,-0.35){its edges alternately};
    \node at (-6,5){\Large $G$};
    \node at (4,5){\Large $G'$};
    \node at (4,-1){\Large $G'$};
    \node at (-6,-1){\Large $G$};
    \node at (10,2.5){$x$};
    \node at (10,-3.5){\color{white}$x$};
    
    \end{tikzpicture}
    \end{center}
    \caption{Procedure \textsc{Degree-Splitting}}
    \label{fig:deg-splitting}
\end{figure}

For $i\in\{-1,1\}$, denote by $\deg_i(v)$ the number of edges incident on $v$ that are colored $i$ in the degree-splitting.
As $|\deg_{-1}(v)-\deg_1(v)|\leq 2$ for any $v\in V$, it follows that
\begin{equation}\label{eq4.1}
    \frac{\deg(v)}{2}-1\leq\deg_i(v)\leq \frac{\deg(v)}{2}+1
\end{equation}
for every $i\in\{-1,1\}$.\\

Using this degree-splitting algorithm we now describe our $(1+\varepsilon)\Delta$-edge-coloring algorithm. The algorithm receives as an input a graph $G=(V,E)$, and a non-negative integer parameter $h$. It first computes a degree-splitting of the input graph $G=(V,E)$, then defines two subgraphs $G_1$ and $G_2$ of $G$, each on the same vertex set $V$, and the set of edges of each of them is defined by the edges that are colored with the same color in the degree-splitting algorithm from the first step. Then, we compute recursively colorings of each of the subgraphs of $G$, and merge these colorings using disjoint palettes. The parameter $h$ determines the depth of the recursion. At the base case of the algorithm ($h=0$) it computes a proper $(\Delta+1)$-edge-coloring using one of the algorithms from Theorem \ref{Delta+1 col}.\\
The pseudo-code of Procedure \textsc{Edge-Coloring} is provided in Algorithm \ref{Procedure Edge-Coloring}. See also Figure \ref{fig:edge-coloring} for an illustration of Procedure \textsc{Edge-Coloring}.

\begin{algorithm}
  \caption{Procedure \textsc{Edge-Coloring} ($G=(V,E),h$)\label{Procedure Edge-Coloring}}
  \begin{algorithmic} [1] 
  \If{$h=0$}
    \LongState{compute a $(\Delta+1)$-edge-coloring $\varphi$ of $G$ using an algorithm based on Theorem \ref{Delta+1 col}}
    \Else
    \Let{$G_1=(V,E_1),G_2=(V,E_2)$}{\textsc{Degree-Splitting}($G$)}
    \LongState{compute a coloring $\varphi_1$ of $G_1$ using Procedure \textsc{Edge-Coloring}($G_1,h-1$)}
    \LongState{compute a coloring $\varphi_2$ of $G_2$ using Procedure \textsc{Edge-Coloring}($G_2,h-1$)}
    \State define coloring $\varphi$ of $G$ by $\varphi(e)=
    \begin{cases}
    \varphi_1(e),& e\in E_1\\
    |\varphi_1|+\varphi_2(e), & e\in E_2 
    \end{cases}$ \myindent{5.5}\Comment{$|\varphi_1|$ is the size of the palette of the coloring $\varphi_1$}
    \EndIf
    \State \textbf{return} $\varphi$
  \end{algorithmic}
\end{algorithm}

\begin{figure}
    \begin{center}
    \addtolength{\leftskip} {-2cm} 
    \addtolength{\rightskip}{-2cm}
    \begin{tikzpicture}

        \begin{scope}[xshift=-7.5cm]            
            \node[draw, ellipse, minimum size = 1.3cm] (nn) at (7,6){$G$};

        \end{scope}
        
        \begin{scope}[xshift=-2.5cm]
            \node[draw, ellipse, minimum size = 0.9cm] (31) at (6,4.5){$G_1$};
            \node[draw, ellipse, minimum size = 0.9cm] (3n) at (8,4.5){$G_2$};
            
            \node[draw, ellipse, minimum size = 1.3cm] (nn) at (7,6){$G$};

            \path [->](nn) edge[ultra thick] (31);
            \path [->](nn) edge[ultra thick] (3n);

        \end{scope}

        \begin{scope}[xshift=2.5cm]
            \node[draw, ellipse, minimum size = 0.9cm] (31) at (6,4.5){$G_1$};
            \node[draw, ellipse, minimum size = 0.9cm] (3n) at (8,4.5){$G_2$};
            
            \node[draw, ellipse, minimum size = 1.3cm] (nn) at (7,6){$G$};
    
            \path [-](nn) edge[ultra thick] (31);
            \path [-](nn) edge[ultra thick] (3n);

            \node at (5.8,4.25) {\color{red}$\bullet$};
            \node at (6,4.25) {\color{green}$\bullet$};
            \node at (6.2,4.25) {\color{blue}$\bullet$};

            \node at (7.8,4.25) {\color{cyan}$\bullet$};
            \node at (8,4.25) {\color{magenta}$\bullet$};
            \node at (8.2,4.25) {\color{gray}$\bullet$};

        \end{scope}

        \begin{scope}[xshift=7.5cm]
            \node[draw, ellipse, minimum size = 0.9cm] (31) at (6,4.5){$G_1$};
            \node[draw, ellipse, minimum size = 0.9cm] (3n) at (8,4.5){$G_2$};
            
            \node[draw, ellipse, minimum size = 1.3cm] (nn) at (7,6){$G$};
    
            \path [->](31) edge[ultra thick] (nn);
            \path [->](3n) edge[ultra thick] (nn);

            \node at (5.8,4.25) {\color{red}$\bullet$};
            \node at (6,4.25) {\color{green}$\bullet$};
            \node at (6.2,4.25) {\color{blue}$\bullet$};

            \node at (7.8,4.25) {\color{cyan}$\bullet$};
            \node at (8,4.25) {\color{magenta}$\bullet$};
            \node at (8.2,4.25) {\color{gray}$\bullet$};
            
            \node at (6.7,5.75) {\color{red}$\bullet$};
            \node at (7,5.75) {\color{green}$\bullet$};
            \node at (7.3,5.75) {\color{blue}$\bullet$};

            \node at (6.7,5.58) {\color{cyan}$\bullet$};
            \node at (7,5.58) {\color{magenta}$\bullet$};
            \node at (7.3,5.58) {\color{gray}$\bullet$};

        \end{scope}

    \draw [-{Stealth[length=5mm]}] (1,6) -- (3,6);
    \node at (2,7.2){partitioning the graph};
    \node at (2,6.8){using Procedure };
    \node at (2,6.4){\textsc{Degree-Splitting}};
    \draw [-{Stealth[length=5mm]}] (5.9,6) -- (7.9,6);
    \node at (6.95,6.8){recursively coloring the};
    \node at (6.95,6.4){subgraphs $G_1$ and $G_2$};
    \draw [-{Stealth[length=5mm]}] (10.8,6) -- (12.8,6);
    \node at (11.9,6.8){merging the coloring};
    \node at (11.9,6.4){using disjoint palettes};

    \end{tikzpicture}
    \end{center}
    \caption{Procedure \textsc{Edge-Coloring}}
    \label{fig:edge-coloring}
\end{figure}
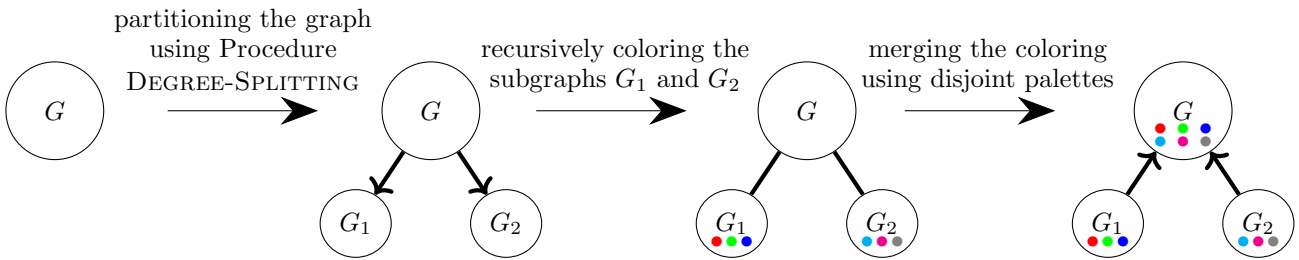

\subsection{Analysis of the Algorithm}\label{sec: Analysis of the Algorithm}

In this section we analyse Procedure \textsc{Edge-Coloring}.\\
We will show that the coloring computed by Procedure \textsc{Edge-Coloring} is a proper edge-coloring. We will bound the number of colors that $\varphi$, defined in Procedure \textsc{Edge-Coloring}, uses. We will also bound the running time of this algorithm in the deterministic and randomized settings, and present a simple trade-off between the number of colors that the algorithm uses and its running time.\\

We begin by presenting some notations that we will use for the analysis. Let $G=(V,E)$ be an input for Procedure \textsc{Edge-Coloring}. We will use the notation $G^{(i)}$ for a graph that is computed after $i$ levels of recursion. We denote its number of vertices by $n^{(i)}$, its number of edges by $m^{(i)}$ and its maximum degree by $\Delta^{(i)}$ (for example $G^{(0)}=G$, $n^{(0)}=n$, $m^{(0)}=m$, and $\Delta^{(0)}=\Delta$).\\

One of the main parameters that affects the performance of our algorithm is the maximum degree of the input 
graph. Therefore, we will start the analysis by bounding the maximum degree of the input graph in every level of the recursion.

\begin{claim}[Maximum degree of recursive graphs]\label{deg bound}
    Let $G=(V,E)$ graph with maximum degree $\Delta$ and $h\leq \log\Delta$ be an input for Procedure \textsc{Edge-Coloring}, and let $0\leq i\leq h$. The maximum degree $\Delta^{(i)}$ of a graph $G^{(i)}$ that was computed after $i$ levels of recursion satisfies $\frac{\Delta}{2^i}-2\leq\Delta^{(i)}\leq\frac{\Delta}{2^i}+2$.
\end{claim}

\begin{proof}
    We prove the claim by induction on $i$:\\
    For $i=0$, we have $G^{(0)}=G$, and indeed the maximum degree of $G$ satisfies $\Delta-2\leq\Delta\leq\Delta+2$.\\
    For $i>0$, let $\Delta^{(i)}$ be the maximum degree of a graph $G^{(i)}$ on which the procedure is invoked on the $i$'th level of the recursion and let $\Delta^{(i-1)}$ be the maximum degree of the graph $G^{(i-1)}$, which we splitted in the previous level of recursion in order to create the graph $G^{(i)}$. By the induction hypothesis and by Equation (\ref{eq4.1}) we have
    $$\Delta^{(i)}\leq\frac{\Delta^{(i-1)}}{2}+1\leq\frac{\frac{\Delta}{2^{i-1}}+2}{2}+1=\frac{\Delta}{2^i}+2,$$
    and
    $$\Delta^{(i)}\geq \frac{\Delta^{(i-1)}}{2}-1\geq \frac{\frac{\Delta}{2^{i-1}}-2}{2}-1= \frac{\Delta}{2^i}-2.$$
\end{proof}

We now bound the size of a subgraph $G^{(i)}$ computed after $i$ levels of recursion. 

\begin{claim}[Recursive graphs' number of edges]\label{edge bound}
    Let $G=(V,E)$ graph with maximum degree $\Delta$ and $h\leq \log\Delta$ be an input for Procedure \textsc{Edge-Coloring}, and let $0\leq i\leq h$. The number of edges $m^{(i)}$ of a graph $G^{(i)}$ that was computed after $i$ levels of recursion satisfies $m^{(i)}\geq\frac{m}{2^i}-2$.
\end{claim}

\begin{proof}
    We prove the claim by induction on $i$:\\
    For $i=0$, we have $G^{(0)}=G$, and indeed the maximum degree of $G$ satisfies $m-2\leq m$.\\
    For $i>0$, let $m^{(i)}$ be the the number of edges of a graph $G^{(i)}$ on which the procedure is invoked on the $i$'th level of the recursion and let $m^{(i-1)}$ be the number of edges of the graph $G^{(i-1)}$, which we splitted in the previous level of recursion in order to create the graph $G^{(i)}$. By the induction hypothesis and by Theorem \ref{Procedure Degree-Splitting} we have
    $$m^{(i)}\geq \left\lfloor\frac{m^{(i-1)}}{2}\right\rfloor \geq \frac{m^{(i-1)}}{2}-1\geq \frac{\frac{m}{2^{i-1}}-2}{2}-1= \frac{m}{2^i}-2.$$
\end{proof}

We next show that the coloring produced by Procedure \textsc{Edge-Coloring} is a proper edge-coloring, and we also bound the number of colors it uses. For this end we will assume that the edge-coloring computed in line 2 of Algorithm \ref{Procedure Edge-Coloring} (base case) is a proper $(\Delta+1)$-edge-coloring. The proof of Lemma \ref{proper edge-coloring} can be found in Appendix \ref{App: Some Proofs from Section 3}.

\begin{lemma}[Proper edge-coloring]\label{proper edge-coloring}
    Let $G=(V,E)$ be a graph with maximum degree $\Delta$.
    Assuming that line 2 of Algorithm \ref{Procedure Edge-Coloring} produces a proper $(\Delta+1)$-edge-coloring, Procedure \textsc{Edge-Coloring} with parameter $h\leq\log\Delta$ computes a proper $\left(\Delta+3\cdot 2^h\right)$-edge-coloring of $G$.
\end{lemma}

We now use this analysis to derive results for deterministic and randomized settings.\\
In our deterministic algorithm, we use the algorithm from Theorem \ref{Delta+1 col}(1) in line 2 of Procedure \textsc{Edge-Coloring}. This algorithm deterministically produces a proper $(\Delta+1)$-edge-coloring. Hence by Lemma \ref{proper edge-coloring} we conclude the following corollary:

\begin{corollary}[Proper edge-coloring - deterministic setting]\label{Proper edge-coloring - deterministic sequential}
    Let $G=(V,E)$ be a graph with maximum degree $\Delta$.
    The deterministic variant of Procedure \textsc{Edge-Coloring} with parameter $h\leq\log\Delta$ computes a proper $\left(\Delta+3\cdot 2^h\right)$-edge-coloring of $G$ in the deterministic setting.
\end{corollary}

For the randomized setting, we now describe the algorithm that we use in line 2 of Procedure \textsc{Edge-Coloring}. In recursive invocations of the algorithm, we apply degree-splittings and continue to the next level without the isolated vertices created by the degree-splitting. As a result, the number $n'$ of vertices in a graph on which the algorithm is invoked recursively may become very small. On the other hand, the failure probability of the randomized algorithm due to~\cite{bernshteyn2023fast} that we use at the bottom level of the recursion, decreases exponentially with $n'$. If $n'$ is too small, this failure probability may be quite large. To address this issue, at the base case of the algorithm (bottom level of the recursion) we will distinguish between two cases:
\begin{itemize}
    \item[-] If the graph $G^{(h)}$ has "many" vertices, namely, if the graph has more than $\tau=(c+1)\cdot\log n$ vertices for some constant $c$, then we will use the randomized algorithm from Theorem \ref{Delta+1 col}(2).
    \item[-] Otherwise, when the graph has a "small" number of vertices, we can compute the edge-coloring deterministically more efficiently, i.e., we will use the deterministic algorithm from Theorem \ref{Delta+1 col}(1), that will require only $O(m^{(h)}\cdot \Delta^{(h)}\cdot\log n^{(h)})=O(m^{(h)}\cdot \Delta^{(h)}\cdot\log\log n)$ time.
\end{itemize}

We now evaluate the probability that the randomized algorithm succeeds and produces the desired edge-coloring. The proof of Claim \ref{Proper edge-coloring - randomized sequential} can be found in Appendix \ref{App: Some Proofs from Section 3}.

\begin{claim}[Proper edge-coloring - randomized setting]\label{Proper edge-coloring - randomized sequential}
    Let $G=(V,E)$ be a graph with $n$ vertices and maximum degree $\Delta$, and let $c>0$ be a constant that determines the threshold value $\tau=(c+1)\cdot\log n$.
    The randomized variant of Procedure \textsc{Edge-Coloring} with parameter $h\leq \log\Delta-2$ computes a proper $\left(\Delta+3\cdot 2^h\right)$-edge-coloring of $G$ with probability at least $1-\frac{1}{n^{c}}$.
\end{claim}

We next proceed to analyse the running time of the algorithm in deterministic and randomized settings.\\

Recall that we assumed that there are no isolated vertices in the graph, i.e., $m\geq \frac{n}{2}$. Hence by Theorem \ref{Delta+1 col}(2) the randomized algorithm requires $O(\Delta^{18}\cdot n)=O(\Delta^{18}\cdot m)$ time.
With this bound, the deterministic and randomized algorithms from Theorem \ref{Delta+1 col} both require $m\cdot F(\Delta,n)$ time for some function $F(\cdot,\cdot)$, monotone in both parameters. The proof of Lemma \ref{time bound} is deferred to Appendix \ref{App: Some Proofs from Section 3}.

\begin{lemma}[A bound on the deterministic and randomized time]\label{time bound}
    Let $G=(V,E)$ be a graph with $n$ vertices, $m$ edges, and maximum degree $\Delta$. If line 2 of the algorithm requires $m\cdot F\left(\Delta,n\right)$ deterministic (respectively, expected) time, then Procedure \textsc{Edge-Coloring} with parameter $h\leq \log\Delta$ requires $O\left(m\cdot F\left(\frac{\Delta}{2^h}+2,n^{(h)}\right)+m\cdot h\right)$ deterministic (resp., expected) time.
\end{lemma}

We are now ready to summarize the main properties of Procedure \textsc{Edge-Coloring} in the deterministic and randomized settings.\\
By Theorem \ref{Delta+1 col}(1), Corollary \ref{Proper edge-coloring - deterministic sequential} and Lemma \ref{time bound}, we get the following properties of the deterministic algorithm:

\begin{corollary}[Properties of deterministic algorithm]\label{deterministic properties}
    Let $G=(V,E)$ be a graph with $n$ vertices, $m$ edges, and maximum degree $\Delta$.
    The deterministic variant of Procedure \textsc{Edge-Coloring} with parameter $h\leq \log
    \Delta$ computes a proper $\left(\Delta+3\cdot 2^h\right)$-edge-coloring of $G$ in $O\left(\frac{m\cdot\Delta\cdot\log n}{2^h}\right)$ time.
\end{corollary}

For the randomized algorithm, let $c>0$ be an arbitrary constant that determines the threshold $\tau=(c+1)\cdot\log n$, and let $t=\frac{m}{(c+1)^2\cdot\log^2n+2}$. Recall that we apply the randomized algorithm from Theorem \ref{Delta+1 col}(2) on the graph $G^{(h)}$ (with $n^{(h)}$ vertices, $m^{(h)}$ edges, and maximum degree $\Delta^{(h)}$) in line 2 of Procedure \textsc{Edge-Coloring} when $n^{(h)}\geq \tau$ (and otherwise we apply the deterministic algorithm from Theorem \ref{Delta+1 col}(1)). Observe that if $\Delta\leq t$, then by Claim \ref{edge bound}, in every subgraph $G^{(h)}$, computed after $h$ levels of recursion, there are at least $m^{(h)}\geq\frac{m}{2^h}-2$ edges. Hence, since $n^{(h)}\geq\sqrt{m^{(h)}}$, and $h\leq \log\Delta\leq\log t$, we get:
$$n^{(h)}\geq\sqrt{m^{(h)}}\geq\sqrt{\frac{m}{2^h}-2}\geq\sqrt{\frac{m}{t}-2}\geq \sqrt{(c+1)^2\cdot\log^2n}=(c+1)\cdot\log n=\tau.$$
Hence, if $\Delta\leq t$, we always apply the randomized algorithm from Theorem \ref{Delta+1 col}(2) in the base case of the algorithm, and the time required for line 2 of Procedure \textsc{Edge-Coloring} on the graph $G^{(h)}$ is $O\left(m^{(h)}\cdot\left(\Delta^{(h)}\right)^{18}\right)$.\\
Otherwise, when $\Delta>t$, observe that for $h\geq\log\left(\Delta\cdot\frac{\log\log n}{4\cdot\log n}\right)$, we have
\begin{eqnarray*}
    h&\geq&\log\left(\Delta\cdot\frac{\log\log n}{4\cdot\log n}\right)>\log\left(t\cdot\frac{\log\log n}{4\cdot\log n}\right)\\
    &=&\log\left(m\cdot\frac{\log\log n}{\left((c+1)^2\cdot\log^2 n +2\right)\cdot4\cdot\log n}\right)=\Omega(\log n),
\end{eqnarray*}
and since, $\frac{\Delta\cdot\log\log n}{2^h}\leq 4\cdot\log n$, we get that $\frac{\Delta\cdot\log\log n}{2^h}=O(h)$.\\
Hence, if $\Delta>t$ and $h\geq\log\left(\Delta\cdot\frac{\log\log n}{4\cdot\log n}\right)$, by Claim \ref{deg bound}, the time required for line 2 of Procedure \textsc{Edge-Coloring} on the graph $G^{(h)}$ is 
\begin{align*}
    O\left(m^{(h)}\cdot\left(\Delta^{(h)}\right)^{18}+m^{(h)}\cdot\Delta^{(h)}\cdot\log n^{(h)}\right)&=O\left(m^{(h)}\cdot\left(\Delta^{(h)}\right)^{18}+\frac{m^{(h)}\cdot\Delta\cdot\log\log n}{2^h}\right)\\
    &=O\left(m^{(h)}\cdot\left(\Delta^{(h)}\right)^{18}+m^{(h)}\cdot h\right).
\end{align*}
Therefore, by Theorem \ref{Delta+1 col}(2), Claim \ref{Proper edge-coloring - randomized sequential} and Lemma \ref{time bound}, we get the following properties:

\begin{corollary}[Properties of randomized algorithm]\label{randomized properties}
    Let $G=(V,E)$ be a graph with $n$ vertices, $m$ edges, and maximum degree $\Delta$, and let $c>0$ be the constant that determines the threshold $\tau=(c+1)\cdot\log n$. The randomized variant of Procedure \textsc{Edge-Coloring} with parameter $\log\left(\Delta\cdot\frac{\log\log n}{4\cdot\log n}\right)\leq h\leq \log\Delta-2$ computes a proper $\left(\Delta+3\cdot 2^h\right)$-edge-coloring of $G$ in 
    $O\left(m\cdot\left(\frac{\Delta}{2^h}\right)^{18}+m\cdot h\right)$ time with probability at least $1-\frac{1}{n^c}$.
\end{corollary}

To get simpler and more useful trade-offs, we set $h\approx\log\left(\varepsilon\cdot\Delta\right)$. Note that if $\varepsilon\leq \frac{\log\log n}{4\cdot\log n}$, then the running time of the randomized algorithm is at least $\Omega\left(m\cdot\varepsilon^{-18}\right)=\Omega\left(\frac{m}{\varepsilon}\cdot\frac{\log^{17}n}{\log^{17}(\log n)}\right)$, which is larger than the running time $O\left(\frac{m\cdot\log n}{\varepsilon}\right)$ of our deterministic algorithm. Hence we use the randomized algorithm only for larger values of $\varepsilon$, i.e., for $\varepsilon\geq\frac{\log\log n}{4\cdot\log n}$. In this range we have $h\approx\log(\varepsilon\cdot\Delta)\geq \log\left(\Delta\cdot\frac{\log\log n}{4\cdot\log n}\right)$.

\begin{theorem}[Trade-off]\label{trade-off}
    Let $G=(V,E)$ be a graph with $n$ vertices, $m$ edges, and maximum degree $\Delta$.
    \begin{enumerate}
        \item[(1)] $\mathrm{(Deterministic}$ $\mathrm{trade-off)}$ Let $\frac{1}{\Delta}\leq\varepsilon<1$. The deterministic variant of Procedure \textsc{Edge-Coloring} with parameter $h=\max\left\{\left\lfloor\log\left(\frac{\varepsilon\cdot\Delta}{3}\right)\right\rfloor,0\right\}$ computes a proper $(1+\varepsilon)\Delta$-edge-coloring of $G$ in $O\left(\frac{m\cdot\log n}{\varepsilon}\right)$ time.
        \item[(2)] $\mathrm{(Randomized}$ $\mathrm{trade-off)}$ Let $\max\left\{\frac{\log\log n}{\log n},\frac{1}{\Delta}\right\}\leq\varepsilon<1$, and let $c>0$ be an arbitrary constant that determines the threshold $\tau=(c+1)\cdot\log n$.
        The randomized variant of Procedure \textsc{Edge-Coloring} with parameter $h=\max\left\{\left\lfloor\log\left(\frac{\varepsilon\cdot\Delta}{4}\right)\right\rfloor,0\right\}$ computes a proper $(1+\varepsilon)\Delta$-edge-coloring of $G$ in 
        $O\left(m\cdot\varepsilon^{-18}+m\cdot \log(\varepsilon\cdot\Delta)\right)$ time with probability at least $1-\frac{1}{n^c}$.
    \end{enumerate}
\end{theorem}

Observe that for any constant value of $\varepsilon>0$, the running time of the randomized variant of the algorithm is $O(m\cdot\log\Delta)$. As long as $\Delta=n^{o(1)}$, it is smaller than the running time $O(m\cdot\log n)$ of the deterministic variant of the algorithm.

\section{Arboricity-Dependent Edge-Coloring}\label{sec: Arboricity-Dependent Edge-Coloring}

In this section we focus on the edge-coloring problem in graphs with bounded \emph{arboricity}. The \textit{arboricity} of a graph $G=(V,E)$ is defined by $\alpha(G)=\max_{S\subseteq V,|S|\geq 2}\left\{\left\lceil\frac{|E(G[S])|}{|S|-1}\right\rceil\right\}$, where $G[S]$ is the subgraph of $G$ induced by the vertices in $S$. Intuitively, we can think of a graph with small arboricity as a graph that is "sparse everywhere". It is well-known~\cite{nash1961decomposition} that for a graph $G=(V,E)$ with arboricity $\alpha$, its edge set can be partitioned into $\alpha$ forests. Such a partition of the edges into forests is called a \emph{forests-decomposition} of the graph. 
The problem of computing a forest-decomposition of the graph into $k$ forests reduces to the problem of computing an orientation of the edge set, such that the outdegree of each vertex is bounded by $k$. (Given such an orientation one can compute a forest-decomposition by partitioning the outgoing edges incident on each vertex among the $k$ forests.) 
It is well-known (see e.g.,~\cite{barenboim2008sublogarithmic}) that a forest-decomposition into $2\alpha$ forests, which also induces an acyclic orientation with outdegree at most $2\alpha$ can be computed in linear time.

The full description of the algorithm that computes such an orientation is provided in Appendix \ref{app: Forests-Decomposition Orientation} for the sake of completeness.

\begin{restatable}[Procedure \textsc{Forests-Decomposition Orientation} properties]{theorem}{ForestsDecompositionOrientation}\label{Forests-Decomposition Orientation properties}
    Let $G=(V,E)$ be an $n$-vertex, $m$-edge graph with arboricity $\alpha$. Procedure \textsc{Forests-Decomposition Orientation} computes a forests-decomposition orientation with degree at most $2\alpha$ in $O(n+m)$ time.
\end{restatable}

\subsection{The Oriented Degree-Splitting Problem}\label{sec: The Oriented Degree-Splitting Problem}

In this subsection we present an algorithm for computing a two-way (or oriented) degree-splitting of a graph equipped with an orientation $\mu$ of all edges. We use this algorithm in our main arboricity-dependent edge-coloring algorithm. We first present the (general) oriented degree-splitting problem.\\

\textbf{The Oriented Degree-Splitting Problem.}
The oriented degree-splitting problem seeks a partitioning of the edge set $E$ of a directed graph $G=(V,E)$ into two parts so that the partition looks almost balanced around each vertex. Concretely, we need to color each edge red or blue such that for each vertex, the difference between the number of \textbf{incoming} red and blue edges incident on it and the difference between the number of \textbf{outgoing} red and blue edges incident on it, are both upper-bounded by some small discrepancy threshold parameter $\kappa$. In other words, we want an assignment $q:E \rightarrow \{1, -1\}$ such that for each vertex $v\in V$, we have $$\left|\sum_{e\in E^{\text{in}}_v} q(e)\right|\leq \kappa \text{ and } \left|\sum_{e\in E^{\text{out}}_v} q(e)\right|\leq \kappa,$$
where $E^{\text{in}}_v$ and $E^{\text{out}}_v$ denote the sets of 
incoming and outgoing edges incident on $v$, respectively. We refer to $\left|\sum_{e\in E^{\text{in}}_v} q(e)\right|$ as the \textit{incoming discrepancy} of $v$, and to $\left|\sum_{e\in E^{\text{out}}_v} q(e)\right|$ as the \textit{outgoing discrepancy} of $v$. We also say that the color red (or 1) is \textit{opposite} to the color blue (or -1). We will also refer to degree-splitting $q(\cdot)$ as above as to a partition $(E_1,E_2)$ of the edge set $E$, where $E_1$ (respectively, $E_2$) is the set of red (respectively, blue) edges.\\

We now present an algorithm that computes an oriented degree-splitting with discrepancy at most 1 that we use in our main edge-coloring algorithm. In our oriented degree-splitting algorithm we use a modified version of \emph{paths decomposition} (see~\cite{israeli1986improved,karloff1987efficient, hanckowiak2001distributed, ghaffari2020improved} for previous work on degree-splitting). In the paths decomposition problem, we are given an undirected graph $G=(V,E)$, and we aim to partition its edge set $E$ into a collection of edge-disjoint paths (not necessarily simple ones), such that each vertex in the graph is an endpoint of at most $\kappa$ such paths, for some small threshold value $\kappa$. Such a partition is called a \emph{paths decomposition} of $G$ with \textit{degree} at most $\kappa$. More specifically, the degree of a vertex $v$ in a paths decomposition is the number of paths such that $v$ is one of their two endpoints.
We refer to a path $P=(v_0,e_1,v_1,...,e_k,v_k)$ as a \textit{maximal path} if all the incident edges of $v_0$ and $v_k$ in the graph are already in $P$, i.e., we cannot add more edges in both endpoints of the path. Note that a maximal path might be a cycle. In the case where $P$ is a cycle, i.e., $v_0=v_k$, we consider the vertex $v_0$ as the \textit{endpoint of the cycle}. Computing a paths decomposition with degree at most 2 can be implemented by a simple linear-time greedy algorithm, that first constructs a maximal path in the graph, then deletes the edges of the path, and recursively computes another maximal path in the remaining graph. Observe that in such paths decomposition, each vertex is an endpoint of at most one path (a vertex might be the two endpoints of a path, in the case of cycle).
Paths decomposition is a useful tool for computing a degree-splitting. This task can be done by coloring the edges of each path using two alternating colors. This coloring defines the degree-splitting of the graph. Below are a few crucial observations that help to prove the correctness of this process:\\
Let $P=(v_0,e_1,v_1,...,e_k,v_k)$ be a path in a paths decomposition.
\begin{itemize}
    \item Every appearance of a vertex $v$ as an inner vertex of a path, i.e., $v=v_i$ for some $i\in\{1,2,...,k-1\}$, does not affect its discrepancy ($e_i$ and $e_{i+1}$ are colored by opposite colors).
    \item Every appearance of a vertex as an endpoint of an odd length cycle increases its discrepancy by 2 ($e_1$ and $e_k$ are colored with the same color).
    \item Every appearance of a vertex as an endpoint of an even length cycle does not affect its discrepancy ($e_1$ and $e_k$ are colored by opposite colors).
    \item Every appearance of a vertex $v$ as an endpoint of a path that is not a cycle changes its discrepancy by 1 (there is one edge incident to this appearance of $v$ in this path).
    \end{itemize}
We will now describe our oriented degree-splitting algorithm. As we mentioned before, this algorithm is an oriented version of the undirected idea that was described above.
To adapt this idea also to oriented graphs, we make one main change. We focus on a certain type of paths, that we call alternating-directions paths.\\

For convenience, we consider a directed graph as an undirected graph $G=(V,E)$ with an orientation $\mu$ on its edges. A path $(v_0,e_1,v_1,...,e_k,v_k)$ in $G$ is called an \emph{alternating-directions path}, or shortly an \emph{AD path} if the orientations (according to $\mu$) of every pair of consecutive edges of the path are alternating, i.e., they are both oriented either towards their common vertex, or from their common vertex. Let $P=(v_0,e_1,v_1,...,e_k,v_k)$ be an AD path. If $e_1$ is oriented towards $v_1$, we refer to $v_0$ as an \emph{outgoing endpoint} of $P$. Otherwise, we refer to $v_0$ as an \emph{incoming endpoint} of $P$. Symmetrically, If $e_k$ is oriented towards $v_{k-1}$, we refer to $v_k$ as an outgoing endpoint of $P$, and otherwise, we refer to $v_k$ as an incoming endpoint of $P$. Note that this is possible that both endpoints of $P$ are incoming or both are outgoing. See Figure \ref{fig:alternating-directions path} for an example of an AD path and its endpoints.

\begin{figure}
    \centering
    \begin{tikzpicture}

        \begin{scope}[yshift=3cm]
        \node at (-0.5,2) {$P_1$};

        \node at (0,0.15) {$v$};
        \node at (11,2.35) {$u$};
        
        \node[circle, fill=black] at (0,0.5) (n1) {};
        \node[circle, fill=black] at (1,2) (n2) {};
        \node[circle, fill=black] at (2,0.5) (n3) {};
        \node[circle, fill=black] at (3,2) (n4) {};
        \node[circle, fill=black] at (4,0.5) (n5) {};
        \node[circle, fill=black] at (5,2) (n6) {};
        \node[circle, fill=black] at (6,0.5) (n7) {};
        \node[circle, fill=black] at (7,2) (n8) {};
        \node[circle, fill=black] at (8,0.5) (n9) {};
        \node[circle, fill=black] at (9,2) (n10) {};
        \node[circle, fill=black] at (10,0.5) (n11) {};
        \node[circle, fill=black] at (11,2) (n12) {};

        \draw[->, line width=0.7mm] (n1)--(n2);
        \draw[->, line width=0.7mm] (n3)--(n2);
        \draw[->, line width=0.7mm] (n3)--(n4);
        \draw[->, line width=0.7mm] (n5)--(n4);
        \draw[->, line width=0.7mm] (n5)--(n6);
        \draw[->, line width=0.7mm] (n7)--(n6);
        \draw[->, line width=0.7mm] (n7)--(n8);
        \draw[->, line width=0.7mm] (n9)--(n8);
        \draw[->, line width=0.7mm] (n9)--(n10);
        \draw[->, line width=0.7mm] (n11)--(n10);
        \draw[->, line width=0.7mm] (n11)--(n12);
        
        \end{scope}

        \node at (-0.5,2) {$P_2$};

        \node at (0,0.15) {$w$};
        \node at (10,0.15) {$z$};
        
        \node[circle, fill=black] at (0,0.5) (n1) {};
        \node[circle, fill=black] at (1,2) (n2) {};
        \node[circle, fill=black] at (2,0.5) (n3) {};
        \node[circle, fill=black] at (3,2) (n4) {};
        \node[circle, fill=black] at (4,0.5) (n5) {};
        \node[circle, fill=black] at (5,2) (n6) {};
        \node[circle, fill=black] at (6,0.5) (n7) {};
        \node[circle, fill=black] at (7,2) (n8) {};
        \node[circle, fill=black] at (8,0.5) (n9) {};
        \node[circle, fill=black] at (9,2) (n10) {};
        \node[circle, fill=black] at (10,0.5) (n11) {};

        \draw[->, line width=0.7mm] (n1)--(n2);
        \draw[->, line width=0.7mm] (n3)--(n2);
        \draw[->, line width=0.7mm] (n3)--(n4);
        \draw[->, line width=0.7mm] (n5)--(n4);
        \draw[->, line width=0.7mm] (n5)--(n6);
        \draw[->, line width=0.7mm] (n7)--(n6);
        \draw[->, line width=0.7mm] (n7)--(n8);
        \draw[->, line width=0.7mm] (n9)--(n8);
        \draw[->, line width=0.7mm] (n9)--(n10);
        \draw[->, line width=0.7mm] (n11)--(n10);

        \node at (5.5,-1.5) {- $w$ and $z$ are both outgoing endpoints of $P_2$};
        \node at (5.5,-0.5) {- $v$ is an outgoing endpoint of $P_1$};
        \node at (5.5,-1) {- $u$ is an incoming endpoint of $P_1$};
    \end{tikzpicture}

    \caption{Alternating-directions paths.}
    \label{fig:alternating-directions path}
\end{figure}
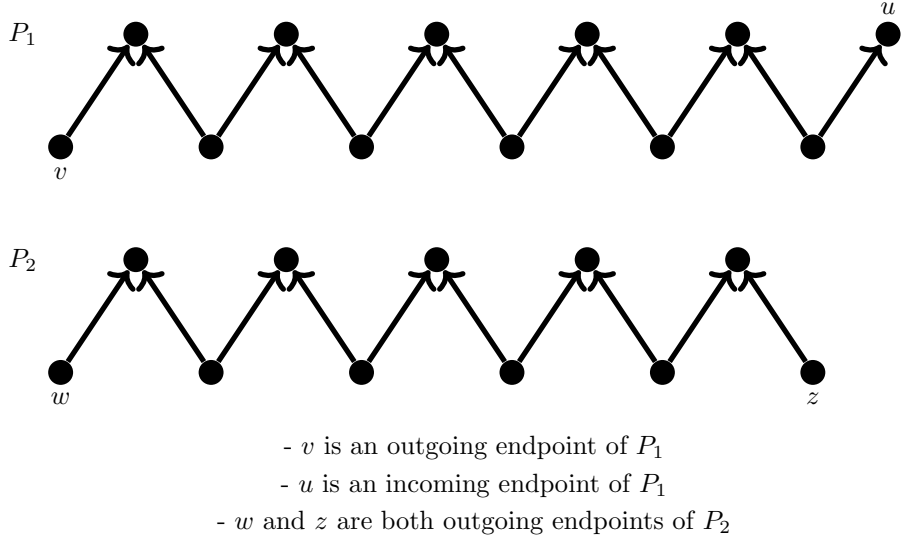

In the alternating-directions paths decomposition problem, we are given an undirected graph $G=(V,E)$ equipped with an orientation $\mu$ on its edges, and we aim to partition its edge set $E$ into a collection of edge-disjoint AD paths (not necessarily simple ones), such that each vertex in the graph is an incoming endpoint of at most $\kappa$ such paths, and an outgoing endpoint of at most $\kappa$ such paths, for some small threshold value $\kappa$. Such a partition is called an \emph{AD paths decomposition} of $(G,\mu)$ with \textit{degree} $\kappa$.\\

Analogously to the undirected case, we refer to an alternating-directions path $P=(v_0,e_1,v_1,...,e_k,$ $v_k)$ as a \textit{maximal AD path} if all the incident edges of $v_0$ that are oriented in the same direction as $e_1$, with respect to $v_0$, are already in $P$, and all the incident edges of $v_k$ that are oriented in the same direction as $e_k$, with respect to $v_k$, are already in $P$, i.e., we cannot add more edges at either endpoint of the path and keep it an alternating-directions path. Note that a maximal AD path might be a cycle. In the case where $P$ is a cycle, i.e., $v_0=v_k$, we refer to the vertex $v_0$ as the \textit{endpoint of the cycle}. Note that if $P$ is an even-length cycle, then every vertex in $P$ can be considered as an endpoint of $v$, while if $P$ is an odd length cycle, only $v_0$ can be considered as its endpoint, since the edges $e_1$ and $e_k$ are oriented in opposite directions with respect to $v_0$ (one edge is oriented towards $v_0$, and the other is oriented from $v_0$) (see Figure \ref{fig:alternating-directions odd cycle} for an example).

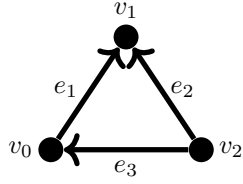
\begin{figure}
    \centering
    \begin{tikzpicture}
        \node[circle, fill=black] at (0,0.5) (n1) {};
        \node[circle, fill=black] at (1,2) (n2) {};
        \node[circle, fill=black] at (2,0.5) (n3) {};
        
        \node at (-0.4,0.5) {$v_0$};
        \node at (1,2.35) {$v_1$};
        \node at (2.4,0.5) {$v_2$};
        
        \node at (0.2,1.3) {$e_1$};
        \node at (1.75,1.3) {$e_2$};
        \node at (1,0.25) {$e_3$};

        \draw[->, line width=0.7mm] (n1)--(n2);
        \draw[->, line width=0.7mm] (n3)--(n2);
        \draw[->, line width=0.7mm] (n3)--(n1);
        
    \end{tikzpicture}

    \caption{An odd-length cycle alternating-directions path $(v_0,v_1,v_2, v_0)$.}
    \label{fig:alternating-directions odd cycle}
\end{figure}

We are now ready to describe the oriented version of the paths decomposition problem. Given a graph $G=(V,E)$ with an orientation $\mu$, we aim to partition the edge set of the graph into a collection of edge-disjoint AD paths (not necessarily simple ones), such that each vertex in the graph is an incoming and an outgoing endpoint in a small number of such paths.\\

Before providing the full description of the AD paths decomposition algorithm, we present some notation. During the algorithm we maintain a set of "active" edges $\mathcal{F}\subseteq E$. At the beginning $\mathcal{F}=E$. We iteratively compute maximal alternating-directions paths, and delete their edges from $\mathcal{F}$. We denote $G_{\mathcal{F}}=(V,\mathcal{F})$. The full description of the algorithm is given in Procedure \textsc{AD Paths Decomposition}:

\begin{algorithm}
  \caption{Procedure \textsc{AD Paths Decomposition} ($G=(V,E),\mu$)\label{Procedure alternating-Directions paths Decomposition}}
  \begin{algorithmic} [1] 
  \Let{$\mathcal{D}$}{$\emptyset$}
  \Let{$\mathcal{F}$}{$E$}
  \While{$\mathcal{F}\neq \emptyset$}
  \State compute a maximal alternating-directions path $P$ in $G_{\mathcal{F}}$ 
  \Let{$\mathcal{D}$}{$\mathcal{D}\cup\{P\}$}
  \Let{$\mathcal{F}$}{$\mathcal{F}\setminus P$\Comment{we delete the edges of $P$}}
  \EndWhile
  \State \textbf{return} $\mathcal{D}$
  \end{algorithmic}
\end{algorithm}

We now analyse the algorithm. We first show that Procedure \textsc{AD Paths Decomposition} computes an alternating-directions paths decomposition $\mathcal{D}$ with degree 1, i.e., each vertex is an incoming and an outgoing endpoint of at most one path in $\mathcal{D}$.

\begin{lemma}[Degree bound for AD paths decomposition]\label{alternating-directions paths decomposition degree bound}
    Let $G=(V,E)$ be an undirected graph and $\mu$ an orientation of its edges. Procedure \textsc{AD Paths Decomposition} partitions the edge set $E$ of the graph into a collection of edge-disjoint AD paths $\mathcal{D}$ such that each vertex $v\in V$ is an incoming and an outgoing endpoint of at most one path in $\mathcal{D}$.
\end{lemma}

\begin{proof}
    First, observe that at the beginning of each iteration of the algorithm, $\mathcal{F}$ consists of all the edges that did not appear in the paths that we have already constructed. Hence, since in each iteration of the algorithm we add to $\mathcal{D}$ an AD directions path in $G_{\mathcal{F}}$, then $\mathcal{D}$ is indeed a partition of $E$ into edge-disjoint AD paths.
    We now show that each vertex $v\in V$ is an incoming endpoint of at most one path in $\mathcal{D}$ (the case of outgoing endpoint is symmetric).
    Assume for a contradiction that there is a vertex $v\in V$ that is an incoming endpoint of at least two paths in $\mathcal{D}$. Denote these paths by $P=(v,e_1,v_1,,...,e_k,v_k)$ and $P'=(v,e'_1,v'_1,...,e'_k,v'_k)$, where $P$ is constructed before $P'$, and $e_1$ and $e_1'$ are both oriented towards $v$. Since $e_1'\in P'$, then $e_1'\notin P$. Hence, since we do not add edges to $\mathcal{F}$, the edge $e_1'$ was in $\mathcal{F}$ at the iteration of the algorithm on which we constructed $P$. Hence, $P$ was not a maximal alternating-directions path, which is a contradiction to the maximality of the path at that stage.
\end{proof}

We next analyse the running time of the algorithm.

\begin{lemma}[A bound on the running time of Procedure \textsc{AD Paths Decomposition}]
    Let $G=(V,E)$ be an undirected graph and $\mu$ an orientation of its edges. Procedure \textsc{AD Paths Decomposition} requires $O(m)$ time.
\end{lemma}

\begin{proof}
    First, observe that the construction of a maximal alternating-directions path $P$ in $G_{\mathcal{F}}$ can be implemented in $O(|P|)$ time. Each vertex can maintain a list of its active incoming edges and a list of its active outgoing edges. Using these lists, a maximal alternating-directions paths can be computed by iteratively choosing the edges of the path, starting at an arbitrary edge $e\in\mathcal{F}$, and extending the path from both endpoints with an active edge oriented in the appropriate direction.
    Hence the overall running time of the algorithm is $\sum_{P\in\mathcal{D}}O(|P|)=O(m)$.
\end{proof}

We summarize this result in the next theorem.

\begin{theorem}[Properties of Procedure \textsc{AD Paths Decomposition}]\label{alternating-directions paths decomposition properties}
    Let $G=(V,E)$ be an $m$-edge undirected graph and $\mu$ an orientation of its edges. Procedure \textsc{AD Paths Decomposition} partitions the edge set $E$ of the graph into a collection of edge-disjoint alternating-directions paths $\mathcal{D}$ such that each vertex $v\in V$ is an incoming and an outgoing endpoint of at most one path in $\mathcal{D}$. The procedure requires $O(m)$ time.
\end{theorem}

Finally, we show how we use an AD paths decomposition in order to construct an oriented degree-splitting. Similarly to the undirected case, we first compute an AD paths decomposition $\mathcal{D}$ of the graph using Procedure \textsc{AD Paths Decomposition}, and then color each path in $\mathcal{D}$ using two alternating colors. We now show that this process produces an oriented degree-splitting of the graph with discrepancy at most 1 in $O(m)$ time.
 
\begin{theorem}[Properties of the oriented degree-splitting algorithm]\label{Directed degree-splitting algorithm properties}
    Let $G=(V,E)$ be an undirected graph and $\mu$ an orientation of its edges.
    Let $\mathcal{D}$ be an AD paths decomposition with degree at most 1. By coloring each path in $\mathcal{D}$ using two alternating colors, we get an oriented degree-splitting $(E_1,E_2)$ of $G$ with discrepancy at most 1.
\end{theorem}

\begin{proof}
    We will show that the incoming discrepancy of each vertex is at most 1 (for outgoing discrepancy the proof is symmetric).
    Denote $\mathcal{D}=\{P_1,P_2,...,P_l\}$, and for each $j\in\{1,2,...,l\}$, denote $P_j=(v_0^j,e_1^j,v_1^j,...,e_{k_j}^j,v_{k_j}^j)$. Let $v\in V$, and denote the set of incoming edges (under $\mu$) incident on $v$ by $E^{\text{in}}_v$. By Theorem \ref{alternating-directions paths decomposition properties}, $v$ is an incoming endpoint of at most one path (or cycle) $P_z\in\mathcal{D}$. 
    We define a set $H^{\text{in}}_v$ of (at most two) edges in the following way: if $v_0^z=v$ and $e_1^z$ is oriented towards $v$, then $e_1^z\in H^{\text{in}}_v$. Also, if $v_{k_z}^z=v$ and $e_{k_z}^z$ is oriented towards $v$, then $e_{k_z}^z\in H^{\text{in}}_v$ too. Note that $H^{\text{in}}_v$ may also be an empty set.
    Let $F^{\text{in}}_v=E^{\text{in}}_v\setminus H^{\text{in}}_v$. We now show that half of the edges in $F^{\text{in}}_v$ are colored by each of the two colors. Observe that all the oriented edges in $F^{\text{in}}_v$ are of the form $\langle v^j_{i-1},v^j_i\rangle$ or $\langle v^j_{i+1},v^j_i\rangle$, where $v=v^j_i$ for $j\in\{1,2,...,l\}$ and $i\in\{1,...,k_j-1\}$. We refer to a pair of edges $\langle v^j_{i-1},v^j_i\rangle$ and $\langle v^j_{i+1},v^j_i\rangle$, where $v=v^j_i$ for $j\in\{1,2,...,l\}$ and $i\in\{1,...,k_j-1\}$, as counterpart edges. Observe that both these edges are in $F^{\text{in}}_v$, and every edge in $F^{\text{in}}_v$ has a unique counterpart in $F^{\text{in}}_v$. Since an edge in $F^{\text{in}}_v$ and its counterpart are consecutive edges in some path of $\mathcal{D}$, they received opposite colors. Hence, half of the edges in $F^{\text{in}}_v$ are colored by each of the two colors. If $\left|H^{\text{in}}_v\right|\leq 1$, we conclude that the incoming discrepancy of $v$ is at most 1. Otherwise, $H^{\text{in}}_v=\left\{e_1^z,e_{k_z}^z\right\}$, which means that $P_z$ is an even-length cycle. Hence $e_1^z$ and $e_{k_z}^z$ received opposite colors, and the incoming discrepancy of $v$ is 0. Overall, the incoming discrepancy of $v$ is at most 1.
\end{proof}

In the next corollary we summarize properties of our degree-splitting algorithm.

\begin{corollary}\label{Directed degree-splitting algorithm properties corollary}
    Let $G=(V,E)$ be an $m$-edge undirected graph and $\mu$ an orientation of its edges.
    Let $\mathcal{D}$ be the AD paths decomposition constructed by Procedure \textsc{AD Paths Decomposition}. By coloring each path in $\mathcal{D}$ using two alternating colors, we get an oriented degree-splitting $(E_1,E_2)$ of $G$ with discrepancy at most 1. The procedure requires $O(m)$ time.
\end{corollary}

\begin{proof}
    By Theorem \ref{alternating-directions paths decomposition properties}, Procedure \textsc{AD Paths Decomposition} computes an AD paths decomposition with degree at most 1. Hence, by Theorem \ref{Directed degree-splitting algorithm properties}, by coloring each path in $\mathcal{D}$ using two alternating colors, we get an oriented degree-splitting of $G$ with discrepancy at most 1.\\

    By Theorem \ref{alternating-directions paths decomposition properties}, the computation of an AD paths decomposition requires $O(m)$ time. The coloring process requires $O(m)$ time as well. Hence the computation of oriented degree-splitting requires $O(m)$ time.
\end{proof}

For $i\in\{-1,1\}$, denote by $\text{indeg}_{i}(v)$ (respectively, $\text{outdeg}_{i}(v)$) the number of incoming (respectively, outgoing) edges incident in $v$ that are colored $i$ in the degree-splitting.
As $|\text{indeg}_{-1}(v)-\text{indeg}_{1}(v)|\leq 1$ and $|\text{outdeg}_{-1}(v)-\text{outdeg}_{1}(v)|\leq 1$, for any $v\in V$, it follows that
\begin{equation}\label{eq4.1 arb}
    \frac{\text{indeg}(v)-1}{2}\leq\text{indeg}_{i}(v)\leq \frac{\text{indeg}(v)+1}{2}
\end{equation}
and
\begin{equation}\label{eq4.2 arb}
    \frac{\text{outdeg}(v)-1}{2}\leq\text{outdeg}_{i}(v)\leq \frac{\text{outdeg}(v)+1}{2},
\end{equation}
for every $i\in\{-1,1\}$.\\

\subsection{An Arboricity-Dependent Edge-Coloring Algorithm}\label{sec: An Arboricity-Dependent Edge-Coloring Algorithm}
In this section we present our arboricity-dependent edge-coloring algorithm.\\
Before describing the algorithm, we first state two known results for $(\Delta+1)$-edge-coloring of graphs with bounded arboricity, due to~\cite{kowalik2024edge,bhattacharya2023density}, which we use as blackbox manner.

\begin{theorem}[$(\Delta+1)$-edge-coloring]\label{Delta+1 col arb}
    Let $G$ be an $n$-vertex $m$-edge graph with maximum degree $\Delta$ and arboricity $\alpha$.
    \begin{enumerate}
        \item[(1)] \cite{kowalik2024edge} There is a deterministic algorithm that computes a proper $(\Delta+1)$-edge-coloring of $G$ in $O(m\cdot\alpha^7\cdot\log n)$ time.
        \item[(2)] \cite{bhattacharya2023density} There is a randomized algorithm that computes a proper $(\Delta+1)$-edge-coloring of $G$ within an expected $O(m\cdot\alpha\cdot\log n)$ time.
    \end{enumerate}
\end{theorem}

Our arboricity-dependent edge-coloring algorithm is an oriented version of Algorithm \ref{Procedure Edge-Coloring} (see Section \ref{sec: Efficient (1+varepsilon)Delta-Edge-Coloring}). Given an $n$-vertex $m$-edge undirected graph $G=(V,E)$, with maximum degree $\Delta$ and arboricity $\alpha$, instead of recursively applying the degree-splitting algorithm (from Section \ref{sec: Efficient (1+varepsilon)Delta-Edge-Coloring}) on the input graph $G$, we first compute a forests-decomposition orientation $\mu$ of $G$, and then apply an oriented version of the edge-coloring algorithm from Section \ref{sec: Efficient (1+varepsilon)Delta-Edge-Coloring}. 
Namely, we are given an input a graph $G=(V,E)$ with an orientation $\mu$ on its edges, and a non-negative integer parameter $h$. The algorithm starts by computing an oriented degree-splitting $(E_1,E_2)$ of $G=(V,E)$ with the orientation $\mu$ (see Theorem \ref{Directed degree-splitting algorithm properties}). Then it defines two subgraphs $G_1$ and $G_2$ of $G$, each on the same vertex set $V$. Their sets of edges are $E_1$ and $E_2$, respectively, in the oriented degree-splitting algorithm from the first step. Then, it recursively computes edge-colorings of each of the subgraphs of $G$, and merges these colorings using disjoint palettes. The parameter $h$ determines the depth of the recursion. At the base case of the algorithm ($h=0$) it computes a proper $(\Delta+1)$-edge-coloring (of the undirected version of the graph) using one of the algorithms from Theorem \ref{Delta+1 col arb}.\\
We present the pseudocode of this algorithm in Procedures \textsc{Oriented Edge-Coloring} and \textsc{Arboricity Edge-Coloring} (Algorithms \ref{Procedure Directed Edge-Coloring} and \ref{Procedure Arboricity Edge-Coloring}).

\begin{algorithm}
  \caption{Procedure \textsc{Oriented Edge-Coloring} ($G=(V,E),\mu,h$)\label{Procedure Directed Edge-Coloring}}
  \begin{algorithmic} [1] 
  \If{$h=0$}
    \LongState{compute a $(\Delta+1)$-edge-coloring $\varphi$ of $G$ using an algorithm from Theorem \ref{Delta+1 col arb}}
    \Else
    \Let{$G_1=(V,E_1),G_2=(V,E_2)$}{\textsc{Oriented Degree-Splitting}($G$)}
    \LongState{compute a coloring $\varphi_1$ of $G_1$ using Procedure \textsc{Oriented Edge-Coloring}($G_1,\mu,h-1$)}
    \LongState{compute a coloring $\varphi_2$ of $G_2$ using Procedure \textsc{Oriented Edge-Coloring}($G_2,\mu,h-1$)}
    \State define coloring $\varphi$ of $G$ by $\varphi(e)=
    \begin{cases}
    \varphi_1(e),& e\in E_1\\
    |\varphi_1|+\varphi_2(e), & e\in E_2 
    \end{cases}$ \myindent{5.5}\Comment{$|\varphi_1|$ is the size of the palette of the coloring $\varphi_1$}
    \EndIf
    \State \textbf{return} $\varphi$
  \end{algorithmic}
\end{algorithm}

\begin{algorithm}
  \caption{Procedure \textsc{Arboricity Edge-Coloring} ($G=(V,E),h$)}\label{Procedure Arboricity Edge-Coloring}
  \begin{algorithmic} [1] 
    \Let{$\mu$}{\textsc{Forests-Decomposition Orientation}($G$)}
    \State \textbf{return} \textsc{Oriented Edge-Coloring}($G$,$\mu$,$h$)
  \end{algorithmic}
\end{algorithm}

We now proceed to the analysis of these algorithms.

We first present some notation that we will use for the analysis. Let $G=(V,E)$ be an $m$-edge graph with maximum (undirected) degree $\Delta$ and arboricity $\alpha$, an orientation $\mu$ of its edges, and $h\leq\log\alpha$ be an input for Procedure \textsc{Oriented Edge-Coloring}, and let $0\leq i\leq h$. We will use the notation $G^{(i)}$ for a graph that is computed after $i$ levels of recursion. We denote its number of edges by $m^{(i)}$, its maximum (undirected) degree by $\Delta^{(i)}$, and its arboricity by $\alpha^{(i)}$. For a vertex $v\in V$, we denote by $\text{indeg}^{(i)}(v)$ and $\text{outdeg}^{(i)}(v)$ the incoming and outgoing degrees of $v$ in such graph $G^{(i)}$, respectively (for example $G^{(0)}=G$, $m^{(0)}=m$, $\Delta^{(0)}=\Delta$, $\alpha^{(0)}=\alpha$, $\text{indeg}^{(0)}(v)=\text{indeg}(v)$ and $\text{outdeg}^{(0)}(v)=\text{outdeg}(v)$, for each $v\in V$). The orientation $\mu$ induces an orientation to every subgraph of $G$. We keep using the notation $\mu$ for the orientation that $\mu$ induces on these subgraphs.\\

Important parameters that affect the performance of our algorithm are the incoming and outgoing degrees of each vertex $v\in V$ on different levels of recursion. Therefore, we start the analysis by bounding these degrees. The analysis is similar to that in Lemma \ref{deg bound}.

\begin{lemma}[Maximum oriented degrees of recursive graphs]\label{directed deg bound}
    Let $G=(V,E)$ be a graph with arboricity $\alpha$, $\mu$ an orientation of its edges with out-degree at most $2\alpha$, and $h\leq\log\alpha$ be an input for Procedure \textsc{Oriented Edge-Coloring}, and let $0\leq i\leq h$. For each $v\in V$, the incoming degree $\mathrm{indeg}^{(i)}(v)$, and the outgoing degree $\mathrm{outdeg}^{(i)}(v)$ of $v$ in a graph $G^{(i)}$ that was computed after $i$ levels of recursion satisfy $$\frac{\mathrm{indeg}(v)}{2^i}-1\leq\mathrm{indeg}^{(i)}(v)\leq\frac{\mathrm{indeg}(v)}{2^i}+1$$ and $$\frac{\mathrm{outdeg}(v)}{2^i}-1\leq\mathrm{outdeg}^{(i)}(v)\leq\frac{\mathrm{outdeg}(v)}{2^i}+1.$$
\end{lemma}

\begin{proof}
    Let $v\in V$. We prove the lemma for the incoming degree of $v$ (the proof of the outgoing degree of $v$ is symmetric).
    The proof is by induction on $i$:\\
    For $i=0$, we have $G^{(0)}=G$, and indeed the incoming degree of $v$ in $G$ satisfies $\mathrm{indeg}(v)-1\leq\mathrm{indeg}(v)\leq\mathrm{indeg}(v)+1$.\\
    For $i>0$, let $G^{(i)}$ be a graph on which the procedure is invoked on the $i$'th level of the recursion, and let $G^{(i-1)}$ be the graph which we splitted in the previous level of recursion in order to create the graph $G^{(i)}$. 
    By the induction hypothesis and by Equation (\ref{eq4.1 arb}) (and Equation (\ref{eq4.2 arb}), for outgoing degree), we have
    $$\mathrm{indeg}^{(i)}(v)\leq\frac{\mathrm{indeg}^{(i-1)}(v)+1}{2}\leq\frac{\frac{\mathrm{indeg}(v)}{2^{i-1}}+1+1}{2}=\frac{\mathrm{indeg}(v)}{2^i}+1,$$
    and
    $$\mathrm{indeg}^{(i)}(v)\geq \frac{\mathrm{indeg}^{(i-1)}(v)-1}{2}\geq \frac{\frac{\mathrm{indeg}(v)}{2^{i-1}}-1-1}{2}= \frac{\mathrm{indeg}(v)}{2^i}-1.$$
\end{proof}

Using this lemma, we now derive a bound on the maximum (undirected) degree $\Delta^{(i)}$ of the subgraphs $G^{(i)}$.

\begin{corollary}[A bound on the maximum (undirected) degree]\label{maximum deg bound}
    Let $G=(V,E)$ be a graph with maximum (undirected) degree $\Delta$ and arboricity $\alpha$, $\mu$ be an orientation of its edges with out-degree at most $2\alpha$, and $h\leq\log\alpha$ be an input for Procedure \textsc{Oriented Edge-Coloring}, and let $0\leq i\leq h$. The maximum (undirected) degree $\Delta^{(i)}$ of (the undirected version of) a graph $G^{(i)}$ that was computed after $i$ levels of recursion satisfies $\Delta^{(i)}\leq\frac{\Delta}{2^i}+2$.
\end{corollary}

\begin{proof}
    The proof is straightforward from Lemma \ref{directed deg bound}. For each $v\in V$, the degrees of $v$ in $G^{(i)}$ and in $G$ satisfy $\deg_{G^{(i)}}(v)=\mathrm{indeg}^{(i)}(v)+\mathrm{outdeg}^{(i)}(v)$, and $\deg_{G}(v)=\mathrm{indeg}(v)+\mathrm{outdeg}(v)$, respectively. Hence, $\deg_{G^{(i)}}(v)\leq\frac{\deg_{G}(v)}{2^i}+2$, which means that $\Delta^{(i)}\leq\frac{\Delta}{2^i}+2$.
\end{proof}

Another important property is that the arboricity, $\alpha^{(i)}$, of the subgraphs $G^{(i)}$, decreases exponentially as well. We prove this in the next lemma.

\begin{lemma}[A bound on the arboricity]\label{arboricity bound}
    Let $G=(V,E)$ be a graph with arboricity $\alpha$ and $h\leq\log\alpha$ be an input for Procedure \textsc{Arboricity Edge-Coloring}, and let $0\leq i\leq h$. The arboricity $\alpha^{(i)}$ of a graph $G^{(i)}$ that was computed after $i$ levels of recursion (of the invocation of Procedure \textsc{Oriented Edge-Coloring} in line 2 of Algorithm \ref{Procedure Arboricity Edge-Coloring}) satisfies $\alpha^{(i)}\leq\frac{\alpha}{2^{i-1}}+1$.
\end{lemma}

\begin{proof}
    By Theorem \ref{Forests-Decomposition Orientation properties}, the maximum outgoing degree of the orientation $\mu$, produced by Procedure \textsc{Forests-Decomposition Orientation} in line 1 of Procedure \textsc{Arboricity Edge-Coloring} (see Algorithm \ref{Procedure Arboricity Edge-Coloring}), is at most $2\alpha$. Hence, by Lemma \ref{directed deg bound}, the outgoing degree, $\mathrm{outdeg}^{(i)}(v)$ of $v$, in a graph $G^{(i)}$ that was computed after $i$ levels of recursion (of Procedure \textsc{Oriented Edge-Coloring} in line 2), satisfies $$\mathrm{outdeg}^{(i)}(v)\leq\frac{\mathrm{outdeg}(v)}{2^i}+1\leq\frac{2\alpha}{2^i}+1=\frac{\alpha}{2^{i-1}}+1.$$ 
    For each vertex $v\in V$, we label its outgoing edges in $G^{(i)}$ using the indexes $\left\{1,2,...,\frac{\alpha}{2^{i-1}}+1\right\}$. Let $(v_1,v_2,...,v_n)$ be the order that the vertices in $G$ become inactive during the execution of Procedure \textsc{Forests-Decomposition Orientation} (vertices become inactive when they leave the set $\mathcal{A}$ of active vertices on line 2 of Algorithm \ref{Procedure Forests-Decomposition Orientation}). Observe that there are no edges of the form $\langle v_i,v_j\rangle$ for $i>j$. Hence the orientation $\mu$ is acyclic. Therefore, for every index $i\in\left\{1,2,...,\frac{\alpha}{2^{i-1}}+1\right\}$, the set of edges labeled $i$ is a forest. Thus, the edge set of $G^{(i)}$ can be split into $\frac{\alpha}{2^{i-1}}+1$ edge-disjoint forests. Hence its arboricity $\alpha^{(i)}$ satisfies $\alpha^{(i)}\leq\frac{\alpha}{2^{i-1}}+1$.
    \end{proof}

We next show that the coloring produced by Procedure \textsc{Oriented Edge-Coloring} is a proper edge-coloring, and bound the number of colors that it uses.
\begin{lemma}[Proper edge-coloring]\label{proper edge-coloring arb}
    Let $G=(V,E)$ be a graph with maximum (undirected) degree $\Delta$ and arboricity $\alpha$, $\mu$ an orientation of its edges with out-degree at most $2\alpha$, and $h\leq\log\alpha$ be an input for Procedure \textsc{Oriented Edge-Coloring}. Procedure \textsc{Oriented Edge-Coloring} computes a proper $\left(\Delta+3\cdot 2^h\right)$-edge-coloring of $G$.
\end{lemma}

By Corollary \ref{maximum deg bound}, the maximal degree of the recursive subgraphs can be bounded by the same bounds as in the undirected algorithm from Section \ref{sec: Analysis of the Algorithm} (Lemma \ref{deg bound}). Hence the proof of Lemma \ref{proper edge-coloring arb} is identical to that of Lemma \ref{deg bound}, and is therefore omitted.\\

Next, we use this lemma to conclude that also Procedure \textsc{Arboricity Edge-Coloring} computes a proper $\left(\Delta+3\cdot 2^h\right)$-edge-coloring of $G$.

\begin{corollary}[Proper edge-coloring]
    Let $G=(V,E)$ be a graph with maximum degree $\Delta$ and arboricity $\alpha$, and $h\leq\log\alpha$ be an input for Procedure \textsc{Arboricity Edge-Coloring}. Procedure \textsc{Arboricity Edge-Coloring} computes a proper $\left(\Delta+3\cdot 2^h\right)$-edge-coloring of $G$.
\end{corollary}

We next proceed to the analysis of the running time of Procedure \textsc{Oriented Edge-Coloring}. Observe that the running time of step 2 of Procedure \textsc{Oriented Edge-Coloring} is of the form $m\cdot F(\alpha,n)$ for some function $F$. Hence by Lemma \ref{time bound} (from Section \ref{sec: Analysis of the Algorithm}, when changing the role of $\Delta$ with $\alpha$), Lemma \ref{arboricity bound}, and by the linearity of the expected value, we conclude the following bound for the (expected) running time of Procedure \textsc{Oriented Edge-Coloring}.

\begin{lemma}[Running time of Procedure \textsc{Oriented-Edge-Coloring}]\label{time bound arb}
    Let $G=(V,E)$ be an $n$-vertex, $m$-edge graph with arboricity $\alpha$, $\mu$ be an orientation of its edges with out-degree at most $2\alpha$, and $h\leq\log\alpha$ be an input for Procedure \textsc{Oriented Edge-Coloring}. If step 2 of Procedure \textsc{Oriented Edge-Coloring} (line 2 of Algorithm \ref{Procedure Directed Edge-Coloring}) requires $O(m\cdot\alpha^c\cdot\log n)$ deterministic (respectively, expected) time for some constant $c>0$, then Procedure \textsc{Oriented Edge-Coloring} requires $O\left(\frac{m\cdot\alpha^c\cdot\log n}{2^{ch}}\right)$ deterministic (resp., expected) time.
\end{lemma}

From this lemma, we now conclude that Procedure \textsc{Arboricity Edge-Coloring} (see Algorithm \ref{Procedure Arboricity Edge-Coloring}) requires $O\left(\frac{m\cdot\alpha^c\cdot\log n}{2^{ch}}\right)$ time, for an appropriate constant $c>0$, as well.

\begin{corollary}[Running time of Procedure \textsc{Arboricity-Edge-Coloring}]
    Let $G=(V,E)$ be an $n$-vertex, $m$-edge graph with arboricity $\alpha$, and $h\leq\log\alpha$ be an input for Procedure \textsc{Arboricity Edge-Coloring}. Then Procedure \textsc{Arboricity Edge-Coloring} requires $O\left(\frac{m\cdot\alpha^7\cdot\log n}{2^{7h}}\right)$ deterministic time. Its randomized variant requires $O\left(\frac{m\cdot\alpha\cdot\log n}{2^{h}}\right)$ expected time.
\end{corollary}

\begin{proof}
    By Theorem \ref{Forests-Decomposition Orientation properties}, step 1 of Procedure \textsc{Arboricity Edge-Coloring} requires $O(n+m)$ deterministic time. For the deterministic algorithm, by Lemma \ref{time bound arb} and Theorem \ref{Delta+1 col arb}(1), step 2 requires $O\left(\frac{m\cdot\alpha^7\cdot\log n}{2^{7h}}\right)$ deterministic time. Hence, since $h\leq\log\alpha$, Procedure \textsc{Arboricity Edge-Coloring} requires $$O\left(\frac{m\cdot\alpha^7\cdot\log n}{2^{7h}}+n+m\right)=O\left(\frac{m\cdot\alpha^7\cdot\log n}{2^{7h}}\right)$$ deterministic time. For the randomized variant of algorithm, by Lemma \ref{arboricity bound}, Lemma \ref{time bound arb}, and Theorem \ref{Delta+1 col arb}(2), the expected running time of step 2 is $O\left(\frac{m\cdot\alpha\cdot\log n}{2^{h}}\right)$ time. Hence, since $h\leq\log\alpha$, the expected running time of Procedure \textsc{Arboricity Edge-Coloring} is $O\left(\frac{m\cdot\alpha\cdot\log n}{2^{h}}+n+m\right)=O\left(\frac{m\cdot\alpha\cdot\log n}{2^{h}}\right)$ time.
\end{proof}

We are now ready to summarize the main properties of Procedure \textsc{Arboricity Edge-Coloring}.\\
\begin{theorem}[Properties of Procedure \textsc{Arboricity Edge-Coloring}]\label{arboricity edge-coloring properties}
    Let $G=(V,E)$ be an $n$-vertex, $m$-edge graph with maximum degree $\Delta$ and arboricity $\alpha$. Procedure \textsc{Arboricity Edge-Coloring} with parameter $h\leq\log\alpha$ computes a proper $\left(\Delta+3\cdot 2^h\right)$-edge-coloring of $G$ in $O\left(\frac{m\cdot\alpha^7\cdot\log n}{2^{7h}}\right)$ deterministic time. Its randomized version requires $O\left(\frac{m\cdot\alpha\cdot\log n}{2^{h}}\right)$ expected time.
\end{theorem}

Note that when $h=0$, the number of colors assigned by the Procedure \textsc{Arboricity Edge-Coloring} is actually $\Delta+1$ (and not $\Delta+3$, as in the bound in Theorem \ref{arboricity edge-coloring properties}).
To get a more simple and useful trade-off, we set $h\approx\log\left(\varepsilon\cdot\alpha\right)$.

\begin{theorem}[Trade-off]\label{trade-off}
    Let $G=(V,E)$ be an $n$-vertex, $m$-edge graph with maximum degree $\Delta$ and arboricity $\alpha$, and let $\frac{1}{\alpha}\leq\varepsilon<1$. Procedure \textsc{Arboricity Edge-Coloring} with parameter $h=\max\left\{\left\lfloor\log\left(\frac{\varepsilon\cdot\alpha}{3}\right)\right\rfloor,0\right\}$ computes a proper $(\Delta+\varepsilon\cdot\alpha)$-edge-coloring of $G$ in $O\left(\frac{m\cdot\log n}{\varepsilon^7}\right)$ deterministic time. Its randomized version requires $O\left(\frac{m\cdot\log n}{\varepsilon}\right)$ expected time.
\end{theorem}

\newpage

\bibliographystyle{alpha}
\bibliography{bibliography.bib}

\newcommand{\etalchar}[1]{$^{#1}$}
\begin{thebibliography}{WHH{\etalchar{+}}97}

\bibitem[Alo03]{alon2003simple}
Noga Alon.
\newblock A simple algorithm for edge-coloring bipartite multigraphs.
\newblock {\em Information Processing Letters}, 85(6):301--302, 2003.

\bibitem[AMSZ03]{aggarwal2003switch}
Gagan Aggarwal, Rajeev Motwani, Devavrat Shah, and An~Zhu.
\newblock Switch scheduling via randomized edge coloring.
\newblock In {\em 44th Annual IEEE Symposium on Foundations of Computer Science, 2003. Proceedings.}, pages 502--512. IEEE, 2003.

\bibitem[Arj82]{arjomandi1982efficient}
Eshrat Arjomandi.
\newblock An efficient algorithm for colouring the edges of a graph with ${\Delta}$+1 colours.
\newblock {\em INFOR: Information Systems and Operational Research}, 20(2):82--101, 1982.

\bibitem[Ass24]{assadi2024faster}
Sepehr Assadi.
\newblock Faster vizing and near-vizing edge coloring algorithms.
\newblock {\em arXiv preprint arXiv:2405.13371}, 2024.

\bibitem[BBKO22]{balliu2022distributed}
Alkida Balliu, Sebastian Brandt, Fabian Kuhn, and Dennis Olivetti.
\newblock Distributed edge coloring in time polylogarithmic in {$\Delta$}.
\newblock In {\em Proceedings of the 2022 ACM Symposium on Principles of Distributed Computing}, pages 15--25, 2022.

\bibitem[BCC{\etalchar{+}}24]{bhattacharya2024faster}
Sayan Bhattacharya, Din Carmon, Mart{\'\i}n Costa, Shay Solomon, and Tianyi Zhang.
\newblock Faster $({\Delta}+ 1)$-edge coloring: Breaking the $m\sqrt{n}$ time barrier.
\newblock {\em arXiv preprint arXiv:2405.15449}, 2024.

\bibitem[BCHN18]{bhattacharya2018dynamic}
Sayan Bhattacharya, Deeparnab Chakrabarty, Monika Henzinger, and Danupon Nanongkai.
\newblock Dynamic algorithms for graph coloring.
\newblock In {\em Proceedings of the Twenty-Ninth Annual ACM-SIAM Symposium on Discrete Algorithms}, pages 1--20. SIAM, 2018.

\bibitem[BCPS23a]{bhattacharya2023arboricity}
Sayan Bhattacharya, Mart{\'\i}n Costa, Nadav Panski, and Shay Solomon.
\newblock Arboricity-dependent algorithms for edge coloring.
\newblock {\em arXiv preprint arXiv:2311.08367}, 2023.

\bibitem[BCPS23b]{bhattacharya2023density}
Sayan Bhattacharya, Mart{\'\i}n Costa, Nadav Panski, and Shay Solomon.
\newblock Density-sensitive algorithms for $({\Delta}+1)$-edge coloring.
\newblock {\em arXiv preprint arXiv:2307.02415}, 2023.

\bibitem[BCPS24]{bhattacharya2024nibbling}
Sayan Bhattacharya, Mart{\'\i}n Costa, Nadav Panski, and Shay Solomon.
\newblock Nibbling at long cycles: Dynamic (and static) edge coloring in optimal time.
\newblock In {\em Proceedings of the 2024 Annual ACM-SIAM Symposium on Discrete Algorithms (SODA)}, pages 3393--3440. SIAM, 2024.

\bibitem[BD23]{bernshteyn2023fast}
Anton Bernshteyn and Abhishek Dhawan.
\newblock Fast algorithms for {V}izing's theorem on bounded degree graphs.
\newblock {\em arXiv preprint arXiv:2303.05408}, 2023.

\bibitem[BE08]{barenboim2008sublogarithmic}
Leonid Barenboim and Michael Elkin.
\newblock Sublogarithmic distributed mis algorithm for sparse graphs using nash-williams decomposition.
\newblock In {\em Proceedings of the twenty-seventh ACM symposium on Principles of Distributed Computing}, pages 25--34, 2008.

\bibitem[BE11]{barenboim2011distributed}
Leonid Barenboim and Michael Elkin.
\newblock Distributed deterministic edge coloring using bounded neighborhood independence.
\newblock In {\em Proceedings of the 30th annual ACM SIGACT-SIGOPS Symposium on Principles of Distributed Computing}, pages 129--138, 2011.

\bibitem[BE13]{barenboim2013basic}
Leonid Barenboim and Michael Elkin.
\newblock Basic distributed graph coloring algorithns.
\newblock In {\em Distributed Graph Coloring: Fundamentals and Recent Developments}, pages 29--45. Springer, 2013.

\bibitem[BEM17]{barenboim2017deterministic}
Leonid Barenboim, Michael Elkin, and Tzalik Maimon.
\newblock Deterministic distributed $({\Delta}+ o ({\Delta}))$-edge-coloring, and vertex-coloring of graphs with bounded diversity.
\newblock In {\em Proceedings of the ACM Symposium on Principles of Distributed Computing}, pages 175--184, 2017.

\bibitem[Ber22]{bernshteyn2022fast}
Anton Bernshteyn.
\newblock A fast distributed algorithm for ({$\Delta$}+1)-edge-coloring.
\newblock {\em Journal of Combinatorial Theory, Series B}, 152:319--352, 2022.

\bibitem[BF81]{beck1981integer}
J{\'o}zsef Beck and Tibor Fiala.
\newblock “{I}nteger-making” theorems.
\newblock {\em Discrete Applied Mathematics}, 3(1):1--8, 1981.

\bibitem[BF20]{blumenstock2020constructive}
Markus Blumenstock and Frank Fischer.
\newblock A constructive arboricity approximation scheme.
\newblock In {\em International Conference on Current Trends in Theory and Practice of Informatics}, pages 51--63. Springer, 2020.

\bibitem[BGW21]{bhattacharya2021online}
Sayan Bhattacharya, Fabrizio Grandoni, and David Wajc.
\newblock Online edge coloring algorithms via the nibble method.
\newblock In {\em Proceedings of the 2021 ACM-SIAM Symposium on Discrete Algorithms (SODA)}, pages 2830--2842. SIAM, 2021.

\bibitem[BKO20]{balliu2020distributed}
Alkida Balliu, Fabian Kuhn, and Dennis Olivetti.
\newblock Distributed edge coloring in time quasi-polylogarithmic in delta.
\newblock In {\em Proceedings of the 39th Symposium on Principles of Distributed Computing}, pages 289--298, 2020.

\bibitem[BMM12]{bahmani2012online}
Bahman Bahmani, Aranyak Mehta, and Rajeev Motwani.
\newblock Online graph edge-coloring in the random-order arrival model.
\newblock {\em Theory of Computing}, 8(1):567--595, 2012.

\bibitem[BNMN92]{bar1992greedy}
Amotz Bar-Noy, Rajeev Motwani, and Joseph Naor.
\newblock The greedy algorithm is optimal for on-line edge coloring.
\newblock {\em Information Processing Letters}, 44(5):251--253, 1992.

\bibitem[Bol98]{bollobas1998modern}
B{\'e}la Bollob{\'a}s.
\newblock {\em Modern graph theory}, volume 184.
\newblock Springer Science \& Business Media, 1998.

\bibitem[Bol12]{bollobas2012graph}
B{\'e}la Bollob{\'a}s.
\newblock {\em Graph theory: an introductory course}, volume~63.
\newblock Springer Science \& Business Media, 2012.

\bibitem[CH82]{cole1982edge}
Richard Cole and John Hopcroft.
\newblock On edge coloring bipartite graphs.
\newblock {\em SIAM Journal on Computing}, 11(3):540--546, 1982.

\bibitem[CK08]{cole2008new}
Richard Cole and {\L}ukasz Kowalik.
\newblock New linear-time algorithms for edge-coloring planar graphs.
\newblock {\em Algorithmica}, 50(3):351--368, 2008.

\bibitem[CN90]{chrobak1990improved}
Marek Chrobak and Takao Nishizeki.
\newblock Improved edge-coloring algorithms for planar graphs.
\newblock {\em Journal of Algorithms}, 11(1):102--116, 1990.

\bibitem[COS01]{cole2001edge}
Richard Cole, Kirstin Ost, and Stefan Schirra.
\newblock Edge-coloring bipartite multigraphs in ${O}({E}\cdot \log{\Delta})$ time.
\newblock {\em Combinatorica}, 21(1):5--12, 2001.

\bibitem[CRV23]{christiansen2023sparsity}
Aleksander~BJ Christiansen, Eva Rotenberg, and Juliette Vlieghe.
\newblock Sparsity-parameterised dynamic edge colouring.
\newblock {\em arXiv preprint arXiv:2311.10616}, 2023.

\bibitem[CY89]{chrobak1989fast}
Marek Chrobak and Moti Yung.
\newblock Fast algorithms for edge-coloring planar graphs.
\newblock {\em Journal of Algorithms}, 10(1):35--51, 1989.

\bibitem[CY07]{cheng2007transmission}
Maggie Cheng and Li~Yin.
\newblock Transmission scheduling in sensor networks via directed edge coloring.
\newblock In {\em 2007 IEEE International Conference on Communications}, pages 3710--3715. IEEE, 2007.

\bibitem[DHZ19]{duan2019dynamic}
Ran Duan, Haoqing He, and Tianyi Zhang.
\newblock Dynamic edge coloring with improved approximation.
\newblock In {\em Proceedings of the Thirtieth Annual ACM-SIAM Symposium on Discrete Algorithms}, pages 1937--1945. SIAM, 2019.

\bibitem[EJ01]{erlebach2001complexity}
Thomas Erlebach and Klaus Jansen.
\newblock The complexity of path coloring and call scheduling.
\newblock {\em Theoretical Computer Science}, 255(1-2):33--50, 2001.

\bibitem[EPS14]{elkin20142delta}
Michael Elkin, Seth Pettie, and Hsin-Hao Su.
\newblock (2{$\Delta$}-1)-edge-coloring is much easier than maximal matching in the distributed setting.
\newblock In {\em Proceedings of the Twenty-Sixth Annual ACM-SIAM Symposium on Discrete Algorithms}, pages 355--370. SIAM, 2014.

\bibitem[FFJ20]{ferber2020towards}
Asaf Ferber, Jacob Fox, and Vishesh Jain.
\newblock Towards the linear arboricity conjecture.
\newblock {\em Journal of Combinatorial Theory, Series B}, 142:56--79, 2020.

\bibitem[Gab76]{gabow1976using}
Harold~N Gabow.
\newblock Using euler partitions to edge color bipartite multigraphs.
\newblock {\em International Journal of Computer \& Information Sciences}, 5(4):345--355, 1976.

\bibitem[GDP08]{gandham2008link}
Shashidhar Gandham, Milind Dawande, and Ravi Prakash.
\newblock Link scheduling in wireless sensor networks: {D}istributed edge-coloring revisited.
\newblock {\em Journal of Parallel and Distributed Computing}, 68(8):1122--1134, 2008.

\bibitem[GHK{\etalchar{+}}20]{ghaffari2020improved}
Mohsen Ghaffari, Juho Hirvonen, Fabian Kuhn, Yannic Maus, Jukka Suomela, and Jara Uitto.
\newblock Improved distributed degree splitting and edge coloring.
\newblock {\em Distributed Computing}, 33(3-4):293--310, 2020.

\bibitem[GK82]{gabow1982algorithms}
Harold~N Gabow and Oded Kariv.
\newblock Algorithms for edge coloring bipartite graphs and multigraphs.
\newblock {\em SIAM journal on Computing}, 11(1):117--129, 1982.

\bibitem[GKMU18]{ghaffari2018deterministic}
Mohsen Ghaffari, Fabian Kuhn, Yannic Maus, and Jara Uitto.
\newblock Deterministic distributed edge-coloring with fewer colors.
\newblock In {\em Proceedings of the 50th Annual ACM SIGACT Symposium on Theory of Computing}, pages 418--430, 2018.

\bibitem[GNK{\etalchar{+}}85]{gabow1985algorithms}
Harold~N Gabow, Takao Nishizeki, Oded Karvin, Daniel Leven, and Osamu Terada.
\newblock Algorithms for edge-coloring graphs.
\newblock {\em Technical Report}, 1985.

\bibitem[GP97]{grable1997nearly}
David~A Grable and Alessandro Panconesi.
\newblock Nearly optimal distributed edge coloring in ${O}(\log\log n)$ rounds.
\newblock {\em Random Structures \& Algorithms}, 10(3):385--405, 1997.

\bibitem[HKP01]{hanckowiak2001distributed}
Michal Hanckowiak, Michal Karonski, and Alessandro Panconesi.
\newblock On the distributed complexity of computing maximal matchings.
\newblock {\em SIAM Journal on Discrete Mathematics}, 15(1):41--57, 2001.

\bibitem[HNS86]{hochbaum1986better}
Dorit~S Hochbaum, Takao Nishizeki, and David~B Shmoys.
\newblock A better than “best possible” algorithm to edge color multigraphs.
\newblock {\em Journal of Algorithms}, 7(1):79--104, 1986.

\bibitem[Hol81]{holyer1981np}
Ian Holyer.
\newblock The np-completeness of edge-coloring.
\newblock {\em SIAM Journal on computing}, 10(4):718--720, 1981.

\bibitem[IS86]{israeli1986improved}
Amos Israeli and Yossi Shiloach.
\newblock An improved parallel algorithm for maximal matching.
\newblock {\em Information Processing Letters}, 22(2):57--60, 1986.

\bibitem[JT11]{jensen2011graph}
Tommy~R Jensen and Bjarne Toft.
\newblock {\em Graph coloring problems}.
\newblock John Wiley \& Sons, 2011.

\bibitem[KLS{\etalchar{+}}22]{kulkarni2022online}
Janardhan Kulkarni, Yang~P Liu, Ashwin Sah, Mehtaab Sawhney, and Jakub Tarnawski.
\newblock Online edge coloring via tree recurrences and correlation decay.
\newblock In {\em Proceedings of the 54th Annual ACM SIGACT Symposium on Theory of Computing}, pages 104--116, 2022.

\bibitem[KN03]{kodialam2003characterizing}
Murali Kodialam and Thyaga Nandagopal.
\newblock Characterizing achievable rates in multi-hop wireless networks: the joint routing and scheduling problem.
\newblock In {\em Proceedings of the 9th annual international conference on Mobile computing and networking}, pages 42--54, 2003.

\bibitem[Kow24]{kowalik2024edge}
{\L}ukasz Kowalik.
\newblock Edge-coloring sparse graphs with ${\Delta}$ colors in quasilinear time.
\newblock {\em arXiv preprint arXiv:2401.13839}, 2024.

\bibitem[KS87]{karloff1987efficient}
Howard~J Karloff and David~B Shmoys.
\newblock Efficient parallel algorithms for edge coloring problems.
\newblock {\em Journal of Algorithms}, 8(1):39--52, 1987.

\bibitem[Lia95]{liang1995fast}
Weifa Liang.
\newblock Fast parallel algorithms for the approximate edge-coloring problem.
\newblock {\em Information processing letters}, 55(6):333--338, 1995.

\bibitem[LSH96]{liang1996parallel}
Weifa Liang, Xiaojun Shen, and Qing Hu.
\newblock Parallel algorithms for the edge-coloring and edge-coloring update problems.
\newblock {\em Journal of Parallel and Distributed Computing}, 32(1):66--73, 1996.

\bibitem[MG92]{misra1992constructive}
Jayadev Misra and David Gries.
\newblock A constructive proof of {V}izing's theorem.
\newblock {\em Information Processing Letters}, 41(3):131--133, 1992.

\bibitem[NK90]{nishizeki19901}
Takao Nishizeki and Kenichi Kashiwagi.
\newblock On the 1.1 edge-coloring of multigraphs.
\newblock {\em SIAM Journal on Discrete Mathematics}, 3(3):391--410, 1990.

\bibitem[NW61]{nash1961decomposition}
C~St~JA Nash-Williams.
\newblock Decomposition of the n-dimensional lattice-graph into hamiltonian lines.
\newblock {\em Proceedings of the Edinburgh Mathematical Society}, 12(3):123--131, 1961.

\bibitem[Sch98]{schrijver1998bipartite}
Alexander Schrijver.
\newblock Bipartite edge coloring in {$O(\Delta m)$} time.
\newblock {\em SIAM Journal on Computing}, 28(3):841--846, 1998.

\bibitem[Sha49]{shannon1949theorem}
Claude~E Shannon.
\newblock A theorem on coloring the lines of a network.
\newblock {\em Journal of Mathematics and Physics}, 28(1-4):148--152, 1949.

\bibitem[Sin19]{sinnamon2019fast}
Corwin Sinnamon.
\newblock Fast and simple edge-coloring algorithms.
\newblock {\em arXiv preprint arXiv:1907.03201}, 2019.

\bibitem[SS05]{sanders2005asymptotic}
Peter Sanders and David Steurer.
\newblock An asymptotic approximation scheme for multigraph edge coloring.
\newblock In {\em SODA}, volume~5, pages 897--906, 2005.

\bibitem[Viz64]{vizing1964estimate}
Vadim~G Vizing.
\newblock On an estimate of the chromatic class of a $p$-graph.
\newblock {\em Discret Analiz}, 3:25--30, 1964.

\bibitem[WHH{\etalchar{+}}97]{williamson1997short}
David~P Williamson, Leslie~A Hall, Jan~A Hoogeveen, Cor~AJ Hurkens, Jan~Karel Lenstra, Sergey~Vasil'evich Sevast'janov, and David~B Shmoys.
\newblock Short shop schedules.
\newblock {\em Operations Research}, 45(2):288--294, 1997.

\end{thebibliography}

\appendix

\section{Forests-Decomposition Orientation}\label{app: Forests-Decomposition Orientation}
In this appendix we present a well known procedure for computing a forests-decomposition into $2\alpha$ forests.\\

Let $G=(V,E)$ be a graph with arboricity $\alpha$.
We begin by presenting the idea behind Procedure \textsc{Forests-Decomposition Orientation}, that computes a forests-decomposition orientation with degree $2\alpha$.
An important observation for this algorithm, is that the average degree of every induced subgraph of $G$ is at most $2\alpha$. Using this property, we can maintain a set "active" vertices $\mathcal{A}$, that consists of the vertices in the graph we did not process yet. At the beginning, $\mathcal{A}=V$, and iteratively we process a vertex $v\in\mathcal{A}$ with the smallest degree in the induced subgraph $G[\mathcal{A}]$. This vertex has degree (in $G[\mathcal{A}]$) at most $2\alpha$, and we orient its incident edges towards theirs other endpoint (note that the edges that are not in $G[\mathcal{A}]$ are already oriented). Before we provide the pseudocode of the algorithm, we present some notations that we use. Let $\mathcal{A}\subseteq V$, and $v\in \mathcal{A}$. We denote by $\deg_{\mathcal{A}}(v)$ the \textit{induced degree} of $v$ in $G[\mathcal{A}]$, and by $N_{\mathcal{A}}(v)$ the \textit{induced neighbors} of $v$ in $G[\mathcal{A}]$. We denote $G[\mathcal{A}]=(\mathcal{A},E[\mathcal{A}])$. The full description of the algorithm is given in Procedure \textsc{Forests-Decomposition Orientation}:

\begin{algorithm}
  \caption{Procedure \textsc{Forests-Decomposition Orientation} ($G=(V,E)$)\label{Procedure Forests-Decomposition Orientation}}
  \begin{algorithmic} [1] 
  \Let{$\mathcal{A}$}{$V$}
  \While{$\mathcal{A}\neq\emptyset$}
  \State let $v\in\mathcal{A}$ with smallest $\deg_{\mathcal{A}}(v)$
  \For{each $u\in N_{\mathcal{A}}(v)$}
  \State $\mu(v,u)=\langle v,u\rangle$
  \EndFor
  \Let{$\mathcal{A}$}{$\mathcal{A}\setminus\{v\}$}
  \EndWhile
  \State \textbf{return} $\mu$
  \end{algorithmic}
\end{algorithm}

We now analyse the algorithm. We first present a lemma that bounds the minimal degree in $G[\mathcal{A}]$.

\begin{lemma}[Minimal degree]\label{minimal degree} Let $G=(V,E)$ a graph with arboricity $\alpha$, and let $\mathcal{A}\subseteq V$. There is a vertex $v\in\mathcal{A}$ such that $\deg_{\mathcal{A}}(v)\leq 2\alpha$.
\end{lemma}

\begin{proof}
    Assume for a contradiction that for each $v\in V$, we have $\deg_{\mathcal{A}}(v)>2\alpha$. Then we get $$2\left|E[\mathcal{A}]\right|=\sum_{v\in \mathcal{A}}\deg_{\mathcal{A}}(v)>|\mathcal{A}|\cdot2\alpha\geq|\mathcal{A}|\cdot 2\cdot\frac{|E[\mathcal{A}]|}{|\mathcal{A}|-1}>2|E[\mathcal{A}]|,$$ which is a contradiction.
\end{proof}

We are now ready to derive the next corollary that bounds the maximal outdegree of the constructed orientation.

\begin{corollary}[A bound on the outdegree]
    Let $G=(V,E)$ be a graph with arboricity $\alpha$. Procedure \textsc{Forests-Decomposition Orientation} produces an orientation $\mu$ of the edge set $E$ with maximum outdegree at most $2\alpha$.
\end{corollary}

\begin{proof}
    Let $v\in V$, and let $\mathcal{A}_v$ be the set $\mathcal{A}$ at the beginning of the iteration where $v$ is chosen in step 3.
    By Lemma \ref{minimal degree}, since $v$ is a vertex in $\mathcal{A}_v$ with minimal degree in $G[\mathcal{A}_v]$, we have $\deg_{\mathcal{A}_v}(v)\leq 2\alpha$. Since all the outgoing edges of $v$ in $\mu$ are defined only in this iteration of the main loop, and since only the incident edges of $v$ in $G[\mathcal{A}_v]$ are oriented from $v$ in this iteration, the outdegree of $v$ in $\mu$ is at most $\deg_{\mathcal{A}_v}(v)\leq 2\alpha$.
\end{proof}

Finally, we present an implementation of the algorithm, and analyse its running time.

\begin{lemma}[A bound on the running time]
    Let $G=(V,E)$ be an $n$-vertex $m$-edge graph. Procedure \textsc{Forests-Decomposition Orientation} requires $O(n+m)$ time.
\end{lemma}

\begin{proof}
    The implementation of this algorithm is done by using an array $A$ of size $n$, where in each entry $i\in\{0,1,2,...,n-1\}$, we store a list of the vertices with degree $i$ in $G[\mathcal{A}]$. During the algorithm, we also maintain for each vertex its degree, $\deg_{\mathcal{A}}(v)$, in $G[\mathcal{A}]$. In addition we save a variable $s$, denoting the first non-empty index of the array. The initialization of $A$, $s$ and the degrees of the vertices requires $O(m)$ time. In each iteration of the algorithm we pick the first vertex $v$ in $A[s]$ (representing a vertex with minimum degree in the subgraph $G[\mathcal{A}]$), orient all the incident edges of $v$ in $G[\mathcal{A}]$. For each neighbor $u\in N_{\mathcal{A}}(v)$, we delete $u$ from $A[\deg_{\mathcal{A}}(u)]$, add it to the list in $A[\deg_{\mathcal{A}}(u)-1]$, and update $\deg_{\mathcal{A}}(u)=\deg_{\mathcal{A}}(u)-1$. 
    Observe that in each iteration, the minimum degree of $G[\mathcal{A}]$ is either decreases by 1, or increases. Hence all the updates of the variable $s$ require $O(m+n)$ time. Overall, this algorithm requires $O(n+m)$ time.
\end{proof}

We summarize this result in the following theorem:

\ForestsDecompositionOrientation*

For our main arboricity-dependent edge-coloring algorithm, that we describe in Section \ref{sec: An Arboricity-Dependent Edge-Coloring Algorithm}, we first compute the above orientation, and then continue the rest of the algorithm on the oriented version of the graph.

\section{Some Proofs from Section \ref{sec: Efficient (1+varepsilon)Delta-Edge-Coloring}}\label{App: Some Proofs from Section 3}
In this appendix we provide some proofs that were omitted from Section \ref{sec: Efficient (1+varepsilon)Delta-Edge-Coloring}.

\begin{proof}[Proof of Theorem \ref{Procedure Degree-Splitting}]
    Without loss of generality we assume that the graph $G$ is connected. Otherwise, we will apply the algorithm on each connected component separately.\\
    Define a graph $G'=(V',E')$, where $V'=V\cup\{x\}$, for a dummy node $x\notin V$, and $E'=E\cup\{(x,v)\;|\;v\in V, \text{ and $\deg_G(v)$ is odd}\}$. Observe that since we increased by 1 only the degree of the odd degree vertices, then for any $v\in V$, its degree in $G'$, $\deg_{G'}(v)$, is even. In addition, since $\sum_{v\in V}\deg_G(v)=2\cdot|E|$, the number of odd degree vertices in $G$ must be even. Hence the degree of $x$ (in $G'$) is even as well. So we get that all the degrees in $G'$ are even, and we can compute an Eulerian cycle in $G'$ using a simple $O(m)$ time algorithm. Next, starting from the dummy node $x$ (if $x$ is isolated in $G'$ we will start from an arbitrary other vertex), we will color the edges of the cycle in two alternating colors, red and blue, except maybe the first and the last edges that we color, that might have the same color (if the cycle is of odd length) in $O(m)$ time. In that way, the discrepancy of each vertex in $G'$ is either 2 (if we started coloring the cycle from this vertex and the cycle is of odd length), or 0 (otherwise).\\
    After removing $x$ (and the edges it is incident on) from $G'$, the discrepancy of its neighbors will become 1, and the discrepancy of the other vertices will not change. So we get a degree-splitting of $G$ with discrepancy at most 2 in $O(m)$ time.\\
    \\Observe that after the coloring of $G'$, without loss of generality, there are $\left\lceil\frac{|E'|}{2}\right\rceil$ edges colored red, and $\left\lfloor\frac{|E'|}{2}\right\rfloor$ edges colored blue. Hence if $x$ is isolated in $G'$, we will get that there are $\left\lceil\frac{m}{2}\right\rceil$ red edges, and $\left\lfloor\frac{m}{2}\right\rfloor$ blue edges. Otherwise, we consider two cases:
    \begin{itemize}
        \item When $m$ is even, we delete the same amount of blue and red edges from $G'$ (when we remove $x$ and its incident edges). Hence we will get that there are $\left\lceil\frac{m}{2}\right\rceil$ red edges, and $\left\lfloor\frac{m}{2}\right\rfloor$ blue edges. 
        \item Otherwise, if $m$ is odd, we delete from $G'$ two more red edges than blue edges (when we remove $x$ and its incident edges). Hence we will get that there are $\left\lceil\frac{m}{2}\right\rceil$ blue edges, and $\left\lfloor\frac{m}{2}\right\rfloor$ red edges in $G$. 
    \end{itemize}
\end{proof}

\begin{proof}[Proof of Lemma \ref{proper edge-coloring}]
    We will start by showing that the coloring $\varphi$ is proper by induction on $h$:\\
    For $h=0$, by the assumption the coloring $\varphi$ is a proper $(\Delta+1)$-edge-coloring.\\
    For $h>0$, let $e_1\neq e_2\in E$. By the induction hypothesis, the colorings $\varphi_1$ and $\varphi_2$ that were defined on lines 5 and 6 of Algorithm \ref{Procedure Edge-Coloring} are proper edge-colorings.\\
    If $e_1,e_2\in E_1$, then $\varphi(e_1)=\varphi_1(e_1)\neq\varphi_1(e_2)=\varphi(e_2)$.\\
    If $e_1,e_2\in E_2$, then $\varphi(e_1)=|\varphi_1|+\varphi_2(e_1)\neq|\varphi_1|+\varphi_2(e_2)=\varphi(e_2)$.\\
    Otherwise, without loss of generality, $e_1\in E_1$, $e_2\in E_2$, and $$\varphi(e_1)=\varphi_1(e_1)\leq|\varphi_1|<|\varphi_1|+\varphi_2(e_2)=\varphi(e_2).$$
    Hence $\varphi$ is a proper edge-coloring.\\
    Next, we analyse the number of colors used by Procedure \textsc{Edge-Coloring}.\\
    Let $G^{(i)}$ be a graph computed after $i$ levels of recursion with maximum degree $\Delta^{(i)}$. 
    Denote by $f(i)$ the number of colors used by the algorithm on the graph $G^{(i)}$. We prove by induction on $i$ that $$f(i)\leq 2^{h-i}\cdot\left(\frac{\Delta}{2^h}+3\right).$$
    For $i=h$, by the assumption the algorithm uses $f(h)=\Delta^{(h)}+1$ colors, and by Claim \ref{deg bound}, $$f(h)=\Delta^{(h)}+1\leq\left(\frac{\Delta}{2^h}+2\right)+1=\frac{\Delta}{2^h}+3.$$
    For $i<h$, we merge two disjoint palettes of size $f(i+1)$ each. Hence, by the induction hypothesis we have $$f(i)= 2\cdot f(i+1)\leq 2\cdot2^{h-i-1}\cdot\left(\frac{\Delta}{2^h}+3\right)=2^{h-i}\cdot\left(\frac{\Delta}{2^h}+3\right).$$
    Therefore, the number of colors employed by the algorithm on the graph $G=G^{(0)}$ is at most $f(0)\leq 2^h\cdot\left(\frac{\Delta}{2^h}+3\right)=\Delta+3\cdot 2^h$.
\end{proof}

\begin{proof}[Proof of Claim \ref{Proper edge-coloring - randomized sequential}]
    First, observe that the only randomized parts of the algorithm are the executions of the algorithm from Theorem \ref{Delta+1 col}(2) in the base cases of the main algorithm (when $h=0$, and there are at least $(c+1)\cdot\log n$ vertices in the graph). 
    Since in each level of recursion we split the graph into two subgraphs and continue with each of them to the next level of recursion, and because we have $h+1$ such levels, we have at most $2^h$ base cases in which we execute the algorithm from Theorem \ref{Delta+1 col}(2). Hence by Lemma \ref{proper edge-coloring}, if all of these executions succeed, then the main algorithm produces a proper $\left(\Delta+3\cdot 2^h\right)$-edge-coloring. We first bound the probability that there exist at least one execution that fails.    
    By Theorem \ref{Delta+1 col}(2), the algorithm from this theorem fails with probability at most $\frac{1}{\left(\Delta^{(h)}\right)^{n^{(h)}}}$. Recall that we apply the algorithm from Theorem \ref{Delta+1 col}(2) when $n^{(h)}\geq (c+1)\cdot\log n$. Hence, by the union bound, by Claim \ref{deg bound}, and since $h\leq \log\Delta-2$, the probability that at least one of the executions fails is at most $$2^h\cdot \frac{1}{\left(\Delta^{(h)}\right)^{n^{(h)}}}\leq\frac{2^h}{\left(\frac{\Delta}{2^h}-2\right)^{(c+1)\cdot\log n}}\leq\frac{2^h}{2^{(c+1)\cdot\log n}}\leq \frac{\Delta}{n^{c+1}}\leq \frac{1}{n^{c}}.$$\\
    Therefore, the probability that the algorithm succeeds is at least $1-\frac{1}{n^{c}}$.
\end{proof}

\begin{proof}[Proof of Lemma \ref{time bound}]
    Let $G^{(i)}$ be a graph with $m^{(i)}$ edges computed after $i$ levels of recursion. By Theorem \ref{Procedure Degree-Splitting}, line 4 and line 7 (of Algorithm \ref{Procedure Edge-Coloring}) require $O\left(m^{(i)}\right)$ time each. Let $c_0$ be the constant such that lines 4 and 7 together require $c_0\cdot m^{(i)}$ time.\\
    Denote by $T\left(i, m^{(i)}\right)$ the running time that Algorithm \ref{Procedure Edge-Coloring} spends on its $i$-level recursive invocation on the subgraph $G^{(i)}$. (Below $n^{(i)}$ is the number of vertices in $G^{(i)}$, $m^{(i)}$ is the number of edges in $G^{(i)}$ and $\Delta^{(i)}$ is its maximum degree.)\\
    We first prove by induction on $i$ that $$T\left(i, m^{(i)}\right)\leq m^{(i)}\cdot F\left(\frac{\Delta}{2^h}+2, n^{(h)}\right) +c_0\cdot (h-i)\cdot m^{(i)}.$$
    For $i=h$, by our assumption and by Claim \ref{deg bound}, the algorithm requires $$T\left(h, m^{(h)}\right)=m^{(h)}\cdot F\left(\Delta^{(h)}, n^{(h)}\right)\leq m^{(h)}\cdot F\left(\frac{\Delta}{2^h}+2, n^{(h)}\right).$$
    For $i<h$, lines 4 and 7 require together $c_0\cdot m^{(i)}$ time. Denote by $m^{(i+1)}_1$ and $m^{(i+1)}_2$ the number of edges in $G_1$ and $G_2$ (defined on line 4), respectively. Then line 5 requires $T\left(i+1,m^{(i+1)}_1\right)$ time and line 6 requires $T\left(i+1,m^{(i+1)}_2\right)$ time. Hence $$T\left(i,m^{(i)}\right)=T\left(i+1,m^{(i+1)}_1\right)+T\left(i+1,m^{(i+1)}_2\right)+c_0\cdot m^{(i)}.$$
    By the induction hypothesis, and since $m^{(i+1)}_1+m^{(i+1)}_2=m^{(i)}$, we have
    \begin{align*}
        T\left(i,m^{(i)}\right)&=T\left(i+1,m^{(i+1)}_1\right)+T\left(i+1,m^{(i+1)}_2\right)+c_0\cdot m^{(i)}\leq\\
        &\leq m^{(i+1)}_1\cdot F\left(\frac{\Delta}{2^h}+2, n^{(h)}\right)+c_0\cdot m^{(i+1)}_1\cdot (h-i-1)+\\        &+m^{(i+1)}_2\cdot F\left(\frac{\Delta}{2^h}+2, n^{(h)}\right)+c_0\cdot m^{(i+1)}_2\cdot (h-i-1)+c_0\cdot m^{(i)}=\\
        &=m^{(i)}\cdot F\left(\frac{\Delta}{2^h}+2, n^{(h)}\right)+c_0\cdot m^{(i)}\cdot(h-i)
    \end{align*}
    Therefore, the algorithm requires $T(0,m)=O\left(m\cdot F\left(\frac{\Delta}{2^h}+2, n^{(h)}\right)+m\cdot h\right)$ time.
\end{proof}

\end{document}